\documentclass[nofootinbib,aps,pra,twocolumn,superscriptaddress]{revtex4-1}

\usepackage{graphicx}  
\usepackage{dcolumn}   
\usepackage{bm}        
\usepackage{amssymb}
\usepackage{physics}   
\usepackage{mathtools}
\usepackage{esvect}

\usepackage{amsthm}
\usepackage{ragged2e}
\usepackage{verbatim}
\usepackage[colorlinks=true,linkcolor=blue,citecolor=red,plainpages=false,pdfpagelabels]{hyperref}
\usepackage{color,xcolor,colortbl}
\usepackage{multirow}
\usepackage{boxedminipage}
\usepackage{changepage}
\usepackage{bbm}
\usepackage{lipsum}
\usepackage{array}
\usepackage{etoolbox}
\usepackage[caption=false]{subfig}
\usepackage{appendix}
\usepackage{float}

\usepackage[mathscr]{euscript}

\usepackage[normalem]{ulem}
\def\>{\rangle}
\def\<{\langle}

\def\({\left(}
\def\){\right)}
\def\[{\left[}
\def\]{\right]}

\newtheorem{observation}{Observation}
\newtheorem{theorem}{Theorem}
\newtheorem{fact}{Fact}

\newtheorem{corollary}{Corollary}

\newtheorem{definition}{Definition}
\newtheorem{example}{Example}

\newtheorem{lemma}{Lemma}

\newtheorem{proposition}{Proposition}

\newcommand{\nc}{\newcommand}
\nc{\rnc}{\renewcommand}
\nc{\be}{\begin{equation}}
\nc{\ee}{\end{equation}}
\nc{\ben}{\begin{eqnarray}}
\nc{\een}{\end{eqnarray}}
\nc{\lbar}[1]{\overline{#1}}

\def\karol{\begin{color}{brown}}
\def\endkarol{\end{color}}
\def\rysiek{\begin{color}{blue}}
 \def\endrysiek{\end{color}\xspace}

\begin{document}

\widetext
%Bounds on the speed of leakage of the generated key and operational approach to Markovianity
\title{Upper bounds on the leakage of private data and operational approach to markovianity}
%Multipartite quantum interactions: privacy criteria and efficacies.
%Fundamental limits on private communication over repeaterless quantum networks
%Fundamental limitations on private capacities
%Fundamental limits on conference key agreement 
%Ultimate limits on private capacities of multiplex quantum channels 
%Ultimate limits on private capacities of multipartite quantum channels 
%\email{khorodec@inf.ug.edu.pl}
\author{Karol Horodecki} \affiliation{Institute of Informatics, National Quantum Information Centre in Gda{\'n}sk, 
Faculty of Mathematics, Physics and Informatics, University of Gda{\'n}sk, 80-952 Gda{\'n}sk, Poland}
\affiliation{International Centre for Theory of Quantum Technologies,
University of Gda{\'n}sk, 80-952 Gda{\'n}sk, Poland}
\author{Micha\l{} Studzi\'nski}\affiliation{Institute of Theoretical Physics and Astrophysics, National Quantum Information Centre in Gda{\'n}sk, Faculty of Mathematics, Physics and Informatics, University of Gda{\'n}sk, 80-952 Gda{\'n}sk, Poland}
\author{Ryszard P. Kostecki}\affiliation{Institute of Informatics, National Quantum Information Centre in Gda{\'n}sk, 
Faculty of Mathematics, Physics and Informatics, University of Gda{\'n}sk, 80-952 Gda{\'n}sk, Poland}
\affiliation{International Centre for Theory of Quantum Technologies,
University of Gda{\'n}sk, 80-952 Gda{\'n}sk, Poland}
\author{Omer Sakarya}\affiliation{Institute of Informatics, National Quantum Information Centre in Gda{\'n}sk, 
Faculty of Mathematics, Physics and Informatics, University of Gda{\'n}sk, 80-952 Gda{\'n}sk, Poland}
\author{Dong Yang}\affiliation{Department of Informatics, University of Bergen, 5020 Bergen, Norway}
\affiliation{Laboratory for Quantum Information, China Jiliang University, 310018 Hangzhou, China }
\date{\today}

\begin{abstract}
We quantify the consequences of a private key leakage and private randomness generated during quantum key distribution. We provide simple lower bounds on the one-way distillable key after the leakage has been detected. We also show that the distributed private randomness does not drop by more than twice the number of qubits of the traced-out system. We further focus on  irreducible private states, showing that their two-way distillable key is  non-lockable.  We then strengthen this result  by referring to the idea of recovery maps. 
We  further consider the  action of special case of side-channels on some of the private states. 
Finally,  we connect the topic of (non)markovian dynamics with that of hacking. In particular, we show that an invertible map is non-CP-divisible if and only if there exists a state whose the key witnessed  by a particular privacy witness increases in time.
This complements the recent result of J. Ko\l{}ody\'nski et al. [Phys. Rev. A 101, 020303(R) (2020)] where the log-negativity was connected with the (non)markovianity of the dynamics.
\end{abstract}

\maketitle

\section{Introduction}
While the security of quantum key distribution is proven in theory, it usually lacks in practice. This is mainly because of (i) the imperfections in the production of the quantum key distribution (QKD) equipment and/or (ii) the active attacks of the eavesdropper known as Trojan Horse attacks (THAs)~\cite{Jain_2014,Sajeed2017}. The latter attacks, such as active inspection of the inner workings of the honest parties' device, can lead to a leakage of the secret key.
Recently there has been taken effort to study the performance of QKD, which takes into account particular examples of the leakages \cite{Wang2018Leaky_source,Wang2021MDI} in the case of quantum key distribution as well as the measurement-device independent quantum key distribution.

In this paper, we consider a more drastic version of THA, according to which eavesdropper gets access to the very {\it raw} key of the honest parties' device. We then note that most of the up-to-date QKD protocols are using in practice one-way communication. (We consider here both device dependent \cite{BB84} and independent \cite{Ekert,BHK} cases, see \cite{RotemPhd} and references therein). Their performance is further based on protocols originating from the Devetak--Winter protocol \cite{DevetakWinter-hash}. We, therefore, focus on the lower bounds on the drop of the {\it raw} key that can be obtained via the latter protocol. It is a practically relevant problem since the raw key should be destroyed properly after key generation. Indeed, the part of the raw key which does not form the final key can be a source of potential leakage and thus should be irreversibly destroyed. Hence, we study how the incorrectly destroyed raw key can influence the security of the key.

Our findings are related to the {\it lockability} of a resource: the problem of how much a given resource drops down under action on (e.g., erasure of) a subsystem of a bipartite quantum state.
There are two variants of the non-lockability of a resource.
According to one of it, the resource should go down by less than the $S(a)$ upon the erasure of system $a$, where $S$ is the von Neumann entropy. We will call it a { \it strong non-lockability}. A weaker version
states that there exists constant $c>0$, independent of the dimension of the state under consideration, such that the resource
does not go down by more than $cS(a)$ (or 
$c\log_2 |a|$). We will call it a
{\it non-lockability}. 

Violation of the strong non-lockability
was proven in \cite{Knig2007} for
the so called {\it accessible information}.
The lockability of entanglement measures
has been first considered in \cite{Horodecki2005}, where entanglement cost $E_C$ \cite{RMPK-quant-ent} was shown to be lockable, while the relative entropy $E_R$ of entanglement was shown to be non-lockable with $c=2$.  In \cite{ChristandlWinter_locking}
lockability of the squashed entanglement $E_{sq}$ was shown.

\paragraph{Motivation}
Before showing the main results, we discuss three possible ways  in which  the eavesdropper can arrange local leakage, which come as a motivation for further studies.

It is known that the eavesdropper can monitor power consumption or the electromagnetic radiation of a working device \cite{Wang2018Leaky_source}. One can also consider a drastic {\it hardware}-THA. Every device which performs quantum key distribution, no matter how shielded, has an incoming fiber. This implies a hole in the shielding. It is then enough to set up a sufficiently strong radioactive source with an open-close mechanism. The bits of generating key can be stored in local memory and further leaked by an open-close mechanism outside via the presence of radiation  ($1$) or lack of it ($0$) in a given slot of time. Monitoring the radioactivity implies directly the leakage of the key. Constant monitoring of radioactivity outside of the device could be a countermeasure to it.

Another attack can be considered in the case of device independent quantum key distribution. It was noticed in 
\cite{Barrett2013} that such a device can be used only once. If used twice, it can leak the key from its previous use by means, e.g., of the accept-abort mechanism. Hence, a device should be destroyed after a single use. This applies not only to the electronics or memory but also to the shielding. This is because shielding can contain a small memory that stores the data. Such an attack can be easily refuted by 
 destroying of the device in the enough irreversible way. 

The easiest way of attack is to set up software that copies the output of the device (a raw key) and distributes it to the eavesdropper. This can be noticed if the system hosting software is constantly monitored. However, noticing the attack does not always mean that it can be stopped, as exemplified by an important variant of this attack: a theft of data. The erasure of classical data happens when the {\it ransomware} (malicious software aimed at ransom) is used by the hackers. Ransomware encrypts the data, which are therefore practically lost unless the (former) owner pays a tribute.

The question is: how much of the security is still at hand after the leakage of the raw key has happened? The bounds obtained in the form of the order of leakage (denoted by constant $c$) considered in the introduction can
help in {\it estimation} of the loss
of data and lead to further shortening of the raw key to obtaining smaller yet still secure key. E.g., in the case of a cloud-storage device exposed for a certain period of time, $\tau$ seconds, to an uncontrolled connection with a certain speed, $v$ megabits per second, one can conclude that no more than $c v \tau$ of megabits were exposed to the attack. (To detect which of the data happened to be copied or erased, one can use the trapdoor mechanism \cite{CantTouchThis}.)

\paragraph{Main results}
\begin{figure}[h]
\centering
 \includegraphics[width=1.0\columnwidth,keepaspectratio,angle=0]{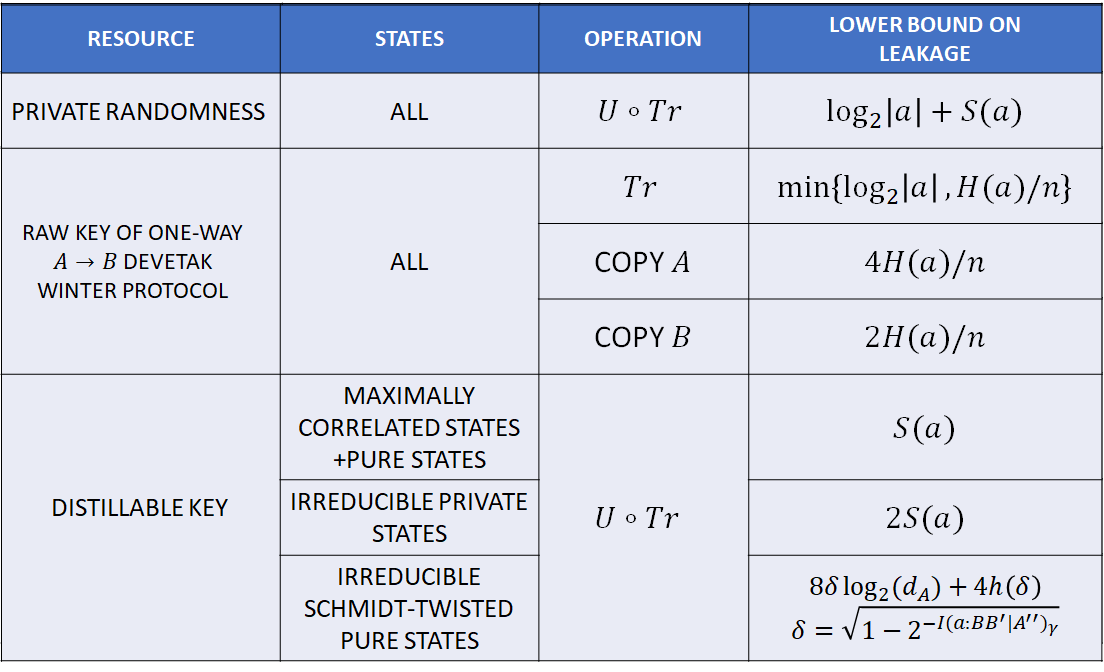}
\caption{Summary of the main results. For either private randomness or private key, and a given class of states, we provide lower bounds on the operation (unitary $U$ composed with partial trace, partial trace, and copying of a system, respectively) on system $a$ with the von-Neumann entropy $S(a)$ and dimension $|a|$. $I(a:BB'|A'')$ is the conditional mutual information.} 
\label{fig:tabelka}
 \end{figure}
 We first consider one of quantum cryptography's fundamental resources, which is the randomness private against a quantum adversary. It is used, e.g., by protocols of generation of the secure key when the honest parties choose settings of measurements (see, e.g., \cite{Bera2017} for review). We focus on a bipartite case introduced in \cite{YHW}. There, two mutually trusting honest parties are distilling private randomness for each of them separately from many copies of a bipartite state $\rho_{AB}$ in the form of an ideal state ${\mathbf{1}_A\over |A|}\otimes{\mathbf{1}_B\over |B|}\otimes \rho_E$, where $\rho_E$ is the purifying system of $\rho_{AB}$. The operations which they use in this resource theory are (i) local unitary operations and (ii) sending quantum states via dephasing channel to the other party (this choice assures that the operations are free, i.e., do not create private randomness). 
 
As the first main result, we show a lower bound
on the drop of private randomness distillable in this scenario. 
Namely, for a bipartite state, under action of (local) unitary transformation followed by partial trace of a subsystem $a$, private randomness does not drop down by more than $S(a)+\log_2|a|$, where $|a|$ is the dimension of $a$.
 In this scenario, one can also consider the rate of randomness obtained without operation (ii) and with or without borrowing local noise \cite{YHW}. Our bound holds in all these cases.

Before turning to the problem of (non)lockability of the key secure against a quantum adversary, let us recall basic facts
about the states containing ideal key, called private states \cite{pptkey,keyhuge}.
A private state has two subsystems: system $AB$ is called the {\it key part} while system $A'B'$ is called a shield \cite{keyhuge}. By definition, one can draw $\log_2 |AB|$ of the key via direct von-Neumann measurement on its key part.
To test how much key drops down for a given private state, we need to control
how much key it contains from the beginning. However, a private state can have
the potentially large key contained in its shielding system $A’B’$. To avoid this problem, we focus on the so-called irreducible private states that have $\log_2 |AB|$ of key - exactly as much as it is directly accessible via the von-Neumann measurement on their key part. 

As the first result related to the secure key, we show
that the key of private states is non-lockable. Precisely, it cannot drop down by erasing system $a$ on one side of it, by more than $2S(a)$.
In that, we partially address the open problem of whether the distillable key can be locked, as presented on the IQOQI list of open problems \cite{IQOQI}.

We then provide first simple bounds on the loss of the {\it raw} one-way distillable key secure against quantum adversary under erasure of data. By one-way distillable key, we mean the one obtained by utilizing one-way classical communication from Alice to Bob. By the raw key, we refer to the key generated via measurement on Alice's side on a quantum state shared by the honest parties.  The raw key then is the bit string that the honest parties share before applying error correction and privacy amplification \cite{RennerPhD}\footnote{We note here that in this article by (ideal) key, we mean the key for the one-time pad, i.e., uniformly random, perfectly correlated pair of bit-strings shared by two honest parties, known only to them. It can be represented by a state $\sum_{i=0}^{d-1} {1\over d}|ii\>\<ii|_{AB}\otimes \rho_E$, where $\rho_E$ represents the total knowledge of the quantum adversary}. 

As one of the main results, we show that the considered type of key is strongly non-lockable (see Theorem \ref{thm:raw_key}). More precisely, it does not drop down by more than $\alpha$
upon the erasure of a system $a$.  Similar results are obtained for the drop of the system at Bob's site: it does not drop down by more than $4\alpha $ upon the erasure of a system $b$ with its entropy scaling with the number of the raw key bits as $n\alpha$.

It is also natural to consider {\it copying} of the data by an adversary, which is a much easier attack than the one described above. In that case, we also observe the non-lockability of the one-way distillable key. It does not drop down by more than $2\alpha$.

Employing simple properties of the {\it smooth min and max entropies}~\cite{RenesRenner2012}, we also provide an alternative lower bound on the drop of the one-way key which reads, in the case considered above, $\log_2 |a|$. 

{\it Bounds on the leakage of two-way distillable key for generalized private states via the fidelity of recovery}. The bounds presented above do not consider
the fact that the system $a$ (or $b$ for Bob) can be almost uncorrelated with the rest of the state of the honest parties. In that case, the drop of the key should be less than the entropy of the copied or erased system. In particular,  when the system $a$ is a product with the rest of the system, the drop of the key should be equal to zero.
To address this case, we use the concept known as {\it fidelity of recovery} \cite{Fawzi2015}, $F_R$. For arbitrary tripartite state $\rho_{aAB}$ fidelity of recovery is the maximum quantum fidelity of $\rho_{aAB}$ with the state $\tilde{\rho}_{AB\tilde{a}}=\Gamma_{A\rightarrow A\tilde{a}}(\rho_{AB})$ recovered by a local quantum map $\Gamma$ acting on system $A$, after erasure of the system $a$. It has been shown \cite{Fawzi2015} that $F_R$ is lower bounded from below by a function $2^{-I(a:B|A)}$, where the conditional mutual information reads  $I(a:B|A):=S(aA)+S(BA)-S(A)-S(ABa)$. While the latter relation is often treated as (in fact, sub-optimal) lower bound on the quantum conditional mutual information, we focus here on the operational meaning of the fidelity of recovery. The conditional mutual information $I(a:B|A)$ quantifies,  to some extent, how much the system $a$ is correlated with the remaining systems. The lower it is, the tighter bound we obtain.

The above relationship allows us to show that the one-way distillable key achieved by i.i.d. operations can not drop down too much if $I(a:B|A)$ is low. By i.i.d., we mean that it is achieved by identical measurement operation and classical pre-processing on each copy of the input state, followed by error correction and privacy amplification \cite{DevetakWinter-hash}.  Although such a quantity may be much lower than the distillable key for a general state, it is equal to the distillable key for   certain generalization of private states called {\it irreducible Shmidt-twisted pure states}. Before stating the results, let us discuss this generalization. A private state can be seen as "twisted" singlet state $|\Psi_+\>$: $\gamma_{ABA'B'}=U|\Psi_+\>\<\Psi_+|\otimes\sigma_{A'B'}U^{\dagger}$, where $\sigma_{A'B'}$ is arbitrary state, and $U=\sum_{i}|ii\>\<ii|\otimes U_i$ is a control unitary transformation called twisting. We generalize this, by inserting a pure state $|\Phi\>$ in place of the singlet,
and allow the unitary $U$ to control the Schmidt basis of $|\Phi\>$ that is a basis in which it can be written as $|\Phi\> =\sum_{i} \sqrt{\lambda_i} |ii\>$. Such obtained state $\gamma'_{ABA'B'}$ we call the {\it irreducible Shmid-twisted pure state}, when $K_D(\gamma')=S(A)_\Phi$ that is the amount of key equals the entropy of the subsystem of the state $|\Phi\>\<\Phi|$.

The following result encapsulates
our findings: 
for any irreducible Schmidt-twisted pure states $\widetilde{\gamma}_{ABA'B'}$, after action $AA'\rightarrow A''a$ and partial trace of system $a$, there is
\begin{equation}
\label{eq:recov}
    K_D(\widetilde{\gamma}_{A''BB'}) 
    \geq K_D(\widetilde{\gamma}_{aA''BB'})- (8 \delta \log_2 d_{A} + 4h(\delta))
\end{equation}
with $\delta =\sqrt{1 - 2^{-I(a:B|A)}}$ (see Proposition \ref{cor:main}), where $h(x)=-x\log_2 x - (1-x)\log_2(1-x)$ is the binary Shannon entropy. Note, that the bound \eqref{eq:recov}  generalizes result for pure states that the key is not lockable (see Theorem \ref{thm:pure}). These and other results are presented in a unified way in Fig. \ref{fig:tabelka}.

{\it Attacks on private states.}
It has been recently proposed \cite{Sakarya2020} that certain private states can serve as a resource for the so called {\it hybrid quantum networks}(a variant of quantum network secure against unauthorized key generation).
Therefore, we also study special attacks on a particular class of private states. We consider several side channels, such as depolarising and amplitude-damping, acting on a shield of a private state. We focus on the private state that can be constructed from an operator $X$ being a (normalised) swap gate (see Eq. (\ref{eq:X_form}) in the Section \ref{sec:notation}). The main insight is that the key drops down by the same amount, no matter how large the system shielding the key is. Therefore, the larger the shield is, the more vulnerable to noise this particular private state becomes.

{\it (Non)markovianity meets hacking.}
We connect two topics, which are usually considered as quite far from each other: the leakage of the private key and the (non)markovianity of quantum dynamics. We consider states of the form $\rho_{ABA'B'}=p_+|\psi_+\rangle\langle\psi_+|_{AB}\otimes\rho_+^{A'B'}+p_-|\psi_-\rangle\langle\psi_-|_{AB}\otimes\rho_-^{A'B'}$, and let $X=\frac{1}{2}(p_+\rho_+^{A'B'}-p_-\rho_-^{A'B'})$. We argue that the distillable key of the so called {\it privacy squeezed state} of $\rho_{ABA'B'}$ exposed to hacking reads
\begin{align}
K_D([\Lambda(\rho_{ABA'B'})]_{psq}) = 1 - h\left({1\over 2}+||(\Lambda_{A'}\otimes {\mathbf{1}}_{B'})X||_1\right),
\label{eq:main}
\end{align}
where $\Lambda(\rho_{ABA'B'})=\Lambda_{A'}\otimes {\mathbf{1}}_{ABB'}(\rho_{ABA'B'})$, and $\Lambda_{A'}$ is a CPTP map acting on the system $A'$ of $\rho_{ABA'B'}$, which corresponds to action of hacking.
Moreover $[\cdot]_{psq}$ is the so called {\it privacy squeezing} \cite{keyhuge} (defined  in Eq. \eqref{rho.psq.definition}).
The privacy squeezing operation is considered here only as a mathematical tool 
rather than a physical map (although it can be physically realised). It allows
to place a lower bound on the distillable key of a given quantum state. Indeed, we have $K_D(\rho)\geq K_D([\rho]_{psq})$ \cite{keyhuge}.
The result presented in Eq. \eqref{eq:main} 
 allows us not only to study the power of leakage  of certain quantum  channels, but also to connect the behaviour of $||X||_1$ due to leakage under hacking with nonmarkovianity of quantum dynamics. (We identify \textit{markovianity} with \textit{CP-divisibility} \cite{Wolf:Cirac:2008,Rivas:Huelga:Plenio:2010}.) Using the results of \cite{Chruciski2011,Chruscinski:Rivas:Stoermer:2018}, and in analogy to \cite{Kolod2020}, we show that the nonmarkovianity of (invertible or image nonincreasing) dynamics, given by a family $\{\Lambda_t\mid t\geq0\}$ of CPTP maps acting on $A'$, is equivalent with
\begin{equation}
    \frac{d}{dt}K_D\left([\Lambda_t(\rho)]_{psq}\right)>0.
\end{equation}

\section{Facts and notations}
\label{sec:notation}
In this section, we invoke important facts and notation used throughout the paper. By $S(\rho_X)$ and $S(\rho_{XY})$, we will mean the von Neumann entropy of systems $X$ and $XY$, respectively. We will also write $S(X)$ and $S(XY)$ if the state is understood from the context. A bipartite state is called a maximally correlated state (MCS) if it is of the form
\be
\label{MCS}
\rho_{AB}=\sum_{i,j}c_{ij}\ketbra{ii}{jj}_{AB},
\ee
where $c_{ij}$ are arbitrary complex numbers.
 The classical-quantum (cq) state is any state of the form
\begin{equation}
    \rho_{cq}=\sum_i p_i |i\>\<i| \otimes \rho_i.
\end{equation}
It is strightforward to check that 
\begin{equation}
    S(\rho_{cq}) =  H(\{p_i\}) + \sum_i p_i S(\rho_i),
\end{equation}
where $H$ denotes Shannon entropy of a distribution $\{p_i\}$.

The private states \cite{pptkey,keyhuge} have the form
\begin{equation}
\gamma_{ABA'B'} = \sum_{i,j}{1\over d} |ii\>\<jj|_{AB}\otimes U_i\sigma U_j^{\dagger},
\label{eq:pstate}
\end{equation}
where $\sigma$ is an arbitrary state on $A'B'$ system.  The private state $\gamma$ is 
called {\it irreducible} if $K_D(\gamma)=\log_2 d$ where $d$ is the dimension of the system $AB$ called the {\it key part}.

The class of irreducible private states is not characterized due to the fact that there can possibly exist states that have zero distillable key but are entangled \cite{Horodecki2018}. Hence we also consider a well characterized, possibly strict subset of irreducible private states, called in \cite{FerraraChristandl} {\it strictly irreducible private states}. The operational meaning of this class is the following. Conditionally on measuring the key part of a strictly irreducible state in a standard basis, there always appears a separable state on their shielding system. Formally, the state \eqref{eq:pstate} is called {\it strictly irreducible} iff the conditional states $U_i\sigma U_i^{\dagger}$ in Eq. (\ref{eq:pstate}) are separable (i.e., they are mixtures of product states) for all $i$. This feature assures that $K_D(\gamma) = \log_2 d$ where
$d$ is the dimension of the key part \cite{karol-PhD}. 
In the case of a private bit, i.e., $d =2$, the private state can be represented by a single operator $X$ with trace norm $||X||_1 = \mathrm{Tr}\sqrt{XX^{\dagger}}$ equal to ${1\over 2}$:
\begin{equation}
\left[\begin{array}{cccc}
\sqrt{XX^{\dagger}} & 0 & 0 & X \\
0 & 0 & 0 & 0 \\
0 & 0 & 0 & 0 \\
X^{\dagger} & 0 & 0 & \sqrt{X^{\dagger}X} 
\end{array}\right].
\label{eq:X_form}
\end{equation}
In \cite{FerraraChristandl} it is shown how to use a one-way Local Operation and Classical Communication to transform any private bit represented by $X$ into
a one represented by {\it hermitian} $\tilde{X}$.
Hence, in our considerations, we can focus on hermitian $X$.

The action of leakage via the map acting on the shielding system
returns the following matrix: 
\begin{equation}
\left[\begin{array}{cccc}
\Lambda_{A'}\otimes {\mathrm I}_{ABB'}\sqrt{XX^{\dagger}} & 0 & 0 & \Lambda_{A'}\otimes {\mathrm I}_{ABB'}X \\
0 & 0 & 0 & 0 \\
0 & 0 & 0 & 0 \\
\Lambda_{A'}\otimes {\mathrm I}_{ABB'}X^{\dagger} & 0 & 0 & \Lambda_{A'}\otimes {\mathrm I}_{ABB'}\sqrt{X^{\dagger}X} 
\end{array}\right].
\end{equation}

To express the connection of the leakage of the key and (non)markovianity
we will need to broaden the class of the interest to states of the form:
\begin{equation}
\rho_{block}:= p_+ |\psi_+\>\<\psi_+|\otimes \rho_+ +p_- |\psi_-\>\<\psi_-|\otimes \rho_-
\label{rhoblock}
\end{equation}
(where $|\psi_\pm\>={1\over \sqrt{2}}(|00\>\pm |11\>$)),
which are private states when $\rho_+ \perp \rho_-$. Following \cite{FerraraChristandl}, we will call them
{\it block states}. An important operation on them is the one that outputs the privacy squeezed state $\rho_{psq}$, i.e.,  the two-qubit bipartite state of the
form
\begin{equation}
\rho_{psq}:=\left[\begin{array}{cccc}
\frac{p_++p_-}{2} & 0 & 0 & \frac{||p_+\rho_+ -p_-\rho_-||_1}{2} \\
0 & 0 & 0 & 0 \\
0 & 0 & 0 & 0 \\
\frac{||p_+\rho_+ -p_-\rho_-||_1}{2} & 0 & 0 & \frac{p_++p_-}{2}
\end{array}\right].
\label{rho.psq.definition}
\end{equation}
There is \cite{keyhuge}:
\begin{equation}
    K_D(\rho_{block}) \geq K_D(\rho_{psq}),
\label{keycomparison}
\end{equation}
where $K_D$ is a key distillable by LOCC operations, defined rigorously in subsection~\ref{def_K_D}.
Due \eqref{keycomparison}, the secure key content of the state $\rho_{psq}$ can be treated as a (non-linear) witness of privacy for the state $\rho$ \cite{Banaszek2012}.

For a given pure state $|\Phi\>_{AB}$ let us consider its Schmidt decomposition $|\Phi\>_{AB}=\sum_{i}\lambda_i|e_i\>\otimes |f_i\>$, where $\lambda_i \geq 0$, and  $\sum_i \lambda_i=1$. A twisting operation in the Schmidt basis of a state $|\Phi\>_{AB}$ is given by
\begin{equation}
\label{Schmidt_twist}
U=\sum_{i,j}|e_if_j\>\<e_if_j|_{AB}\otimes U^{(ij)}_{A'B'},
\end{equation}
where for each $(ij)$, $U^{(ij)}_{A'B'}$ is some unitary operation. This leads to a concept of the \textit{Schmidt-twisted pure state} $\widetilde{\gamma}_{ABA'B'}$, which is defined as
\begin{equation}
\label{Schmidt_pure}
\begin{split}
\widetilde{\gamma}_{ABA'B'}&:= U\left(|\Phi\>\<\Phi|_{AB}\otimes \sigma_{A'B'}\right)U^{\dagger}\\ &=\sum_{i,j}\lambda_i\lambda_j|e_if_i\>\<e_jf_j|\otimes U_i\sigma U_j^{\dagger},
\end{split}
\end{equation}
where $\sigma$ is defined on systems $A'$ and $B'$ (for clarity, we suppressed subsystem indices).  The Schmidt-twisted pure state $\widetilde{\gamma}_{ABA'B'}$ is called irreducible if
it satisfies $K_D(\widetilde{\gamma}_{ABA'B'}) = S(A)_{\Phi}$. This means that its whole security content is accessible by a direct von Neumann measurement on its key part system $AB$.

Finally, for self-consistence of this manuscript we define the Uhlmann fidelity~\cite{Uhlmann:1976,RJozsa} for two quantum states $\rho$ and $\sigma$:
\begin{equation}
F(\rho,\sigma):=\left(\operatorname{tr}\sqrt{\sqrt{\rho}\sigma\sqrt{\rho}}\right)^2 .
\end{equation}
This expression can be written in equivalent form $||\sqrt{\rho}\sqrt{\sigma}||_1^2$, where $||\cdot||_1$ denotes trace norm.

\subsection{Entanglement measures}
Here, we introduce entanglement measures that are employed in this manuscript - the relative entropy of entanglement, distillable entanglement, and squashed entanglement. 
\begin{definition}
	\label{Er}
	The relative entropy of entanglement for an arbitrary density operator $\rho $ is defined as
	\be
	E_R(\rho):=\mathop{\inf}\limits_{\omega \in \mathcal{SEP}} D(\rho|\omega),
	\ee
	where the infimum runs over the set of separable states $\mathcal{SEP}$, and $D(\cdot|\cdot)$ denotes relative entropy, i.e. $ D(\rho|\sigma):= \tr\rho \log_2 \rho - \tr \rho \log_2 \sigma$, for an arbitrary density operators $\rho,\sigma$.
\end{definition}

\begin{definition}
	For all bipartite states $\rho_{AB}$ we define one-way distillable entanglement
	\be
	E_D^{\rightarrow}:=\lim_{\epsilon \rightarrow 0}\lim_{n\rightarrow \infty}\sup_{\Lambda_{A\rightarrow B}}\left\lbrace E: \Lambda \left(\rho^{\otimes n} \right)\approx_{\epsilon}\Phi_{AB}(2^{nE})  \right\rbrace, 
	\ee
	where maps $\Lambda_{A\rightarrow B}$ are restricted to one-way $LOCC$, and $\Phi_{AB}(2^{nE})$ is maximally entangled state between $A$ and $B$ of Schmidt rank $2^{nE}$.
\end{definition}
In the above expressions we use the notation $\rho \approx_{\epsilon} \sigma$ for $\left| \left|\rho-\sigma \right| \right|_1\leq \epsilon$ to compress the definitions.

\begin{definition}
	\label{Isq}
	The squashed entanglement \cite{Squashed} for an arbitrary bipartite sate $\rho_{AB}$ is defined as
	\be
	E_{sq}\left(\rho_{AB} \right):= \mathop{\inf}\limits_{\rho_{ABE}}\left\lbrace \frac{1}{2}I(A;B|E) \ | \ \rho_{ABE} \ \text{extension of} \ \rho_{AB} \right\rbrace.
	\ee
	The infimum is taken over all extensions of $\rho_{AB}$, i.e. over all density operators $\rho_{ABE}$ with $\rho_{AB}=\tr_E\rho_{ABE}$. By $I(A;B|E):= S(AE)+S(BE)-S(ABE)-S(E)$ we denote the quantum conditional mutual information of $\rho_{ABE}$ \cite{CerfAdami}. $S(A):=S(\rho_A)$ is the von Neumann entropy of the underlying state.
\end{definition}

\subsection{Min- and max- entropies, and their smoothed versions}
\label{subMinMax}
We begin from defining the min- and max- entropies (see  see~\cite{RennerPhD,RenesRenner2012,Tomamichel_2016} for the details).
For a given bipartite state $\rho_{AB}$, they are given by
\begin{equation}
\label{minmaxE}
\begin{split}
&H_{\min}(A|B)_{\rho}:=\sup_{\sigma_B}\sup\left\{\lambda \in\mathbb{R}:\rho_{AB}\leq 2^{-\lambda}\mathbf{1}_A\otimes \sigma_B \right\},\\
&H_{\max}(A|B)_{\rho}:=\max_{\sigma_B}\operatorname{log}_2F(\rho_{AB},\mathbf{1}_A\otimes \sigma_B),
\end{split}
\end{equation}
where $F(\rho,\sigma)=||\sqrt{\rho}\sqrt{\sigma}||_1$ denotes fidelity between quantum states $\rho$ and $\sigma$.  
The $\epsilon$-smooth min- and max- entropies of $A$ conditioned on $B$ of the state $\rho_{AB}$ read, respectively,
\begin{equation}
\label{Hminmaxsmooth}
\begin{split}
&H^{\epsilon}_{\operatorname{min}}(A|B)_{\rho}:=\max_{\widetilde{\rho}_{AB}\in \mathcal{B}^{\epsilon}(\rho_{AB})}H_{\min}(A|B)_{\widetilde{\rho}},\\
&H^{\epsilon}_{\operatorname{max}}(A|B)_{\rho}:=\min_{\widetilde{\rho}_{AB}\in \mathcal{B}^{\epsilon}(\rho_{AB})}H_{\max}(A|B)_{\widetilde{\rho}},
\end{split}
\end{equation}
where $\mathcal{B}^{\epsilon}(\rho_{AB})$ is $\epsilon$-ball of states $\widetilde{\rho}_{AB}$ which are $\epsilon$-close to $\rho_{AB}$.
For the further considerations let us also remind here that the smooth entropies of the i.i.d. product state $\rho_{A^nB^n}=\rho_{AB}^{\otimes n}$ converge to conditional Shannon entropy $H_{\rho}(A|B)$ for $n\rightarrow \infty$. More precisely, 
\begin{equation}
\label{minmaxconv}
\begin{split}
&\lim_{n\rightarrow \infty}\left\{\frac{1}{n}H^{\epsilon}_{\operatorname{min}}(A^n|B^n)_{\rho}\right\}\geq H(A|B)_{\rho},\\
&\lim_{n\rightarrow \infty}\left\{\frac{1}{n}H^{\epsilon}_{\operatorname{max}}(A^n|B^n)_{\rho}\right\}\leq H(A|B)_{\rho}.
\end{split}
\end{equation}

\subsection{Key distillable by LOPC operations}
\label{sec:LOPC}
For further purposes, we remind here the idea of the LOPC, Local (quantum) Operations, and Public Classical Communication, with corresponding distillable key $C_D$ for tripartite quantum state $\rho=\rho_{ABE}$. In this scenario, three parties, Alice, Bob, and Eve, hold many systems in the same tripartite state $\rho$. Alice and Bob can process input states by quantum operations, each in their respective laboratory, and they communicate publicly classical messages, with copies also sent to eavesdropper Eve. For a more formal definition of LOPC operations, see Definition 4.2 in~\cite{karol-PhD,Ekert0}.  Historically,
its one-way version was defined first in \cite{DevetakWinter-hash}, in a way equivalent to the following one, where $\rho\equiv \rho_{ABE}$
\begin{equation}
 K_\rightarrow(\rho):=\inf_{\epsilon >0}\limsup_{n\rightarrow \infty}\sup_{\Delta \in LOPC_\rightarrow}\left\{\frac{\log_2 d}{n} \ | \  \Delta(\rho^{\otimes n })\approx_{\epsilon} \tau_d\right\},
\end{equation}
where $LOPC_\rightarrow$ denotes the LOPC operations, in which the classical communication goes from $A$ to $B$ only, while $\tau_d=(1/d)\sum_{i=1}^{d-1}|ii\>\<ii|\otimes \rho_E$ is a ccq-state with $\log_2 d$ secure bits, and $\approx_{\epsilon}$ denotes $\epsilon$-closeness in the trace norm $||\cdot||_1$.

The (two-way) distillable classical key between Alice and Bob from a quantum tripartite state $\rho\equiv \rho_{ABE}$ utilizing LOPC operations is given as~\cite{Ekert0}:
\begin{equation}
 C_D(\rho):=\inf_{\epsilon >0}\limsup_{n\rightarrow \infty}\sup_{\Delta \in LOPC}\left\{\frac{\log_2 d}{n} \ | \  \Delta(\rho^{\otimes n })\approx_{\epsilon} \tau_d\right\}.
\end{equation}

There is no closed formula known for $C_D$ for a general state.
However, when one restricts the one-way LOCC communication in the distillation process, then there is a formula for the distillable key, called a one-way distillable key, given by Devetak and Winter \cite{DevetakWinter-hash}. We invoke here the theorem
which encapsulates this rather complicated formula.

\begin{theorem} [\cite{DevetakWinter-hash}, in formulation of \cite{Nowakowski2016}]
\label{DW_Nowakowski}
For every state $\rho_{ABE}$, $K_{\rightarrow} =
\lim_{n\rightarrow \infty}
{K^{(1)}(\rho^{\otimes n})\over n}$, with $K^{(1)}= \max_{Q;T|X}(I(X :
B|T) -I(X : E|T))$, where the maximization is over
all POVMs $Q = (Q_x)_{x\in{\cal X}}$ and channels R such that
$T = R(X)$, while the information quantities refer to the
state $\omega_{TABE} =
\sum_{t,x} R(t|x)P(x) |t\>\<t| \otimes |x\>\<x|\otimes Tr_A(\rho_{ABE}(Q_x)\otimes {\mathbf{1}}_{BE})$: The range of the measurement Q
and the random variable $T$ may be assumed to be bounded
as follows: $|T| \leq d^2_A
$ and $|X| \leq d^2_A$
 where $T$ can be taken
as a (deterministic) function of $X$.
\end{theorem}

We have then, by definition that $K_{\rightarrow}(\rho_{ABE}) \leq C_D(\rho_{ABE})$, for any tripartite state $\rho_{ABE}$. This is by the fact that the class of protocols in definition of $K_{\rightarrow}$ is strictly less than in the case of $C_D$. In what follows, we will need a lower bound
on $K_{\rightarrow}$, which bases on restricting operations in its definition
to be identical on each copy of $\rho_{ABE}$. Namely we define the one-way i.i.d. version of a one-way secure key, $K^{iid}$, with $\{Q_x\}_{x\in \mathcal{X}}$ in the form ${\hat{Q}_x}^{\otimes n}$ and $T$ in the form $\hat{T}^{\times n}$. That is, the measurement on Alice's side is performed 
identically and independently on each copy
of the state, rather than globally, and further classical information comes from a variable $\hat{T}$ that is identical on each copy:
\begin{definition}
\label{key_iid}
For every state $\rho_{ABE}$, a one-way i.i.d. secure key reads 
\begin{equation}
\label{eq_key_iid}
K^{iid}(\rho_{ABE})= \lim_{n\rightarrow \infty} \frac{1}{n} \max_{\hat{Q};\hat{T}|X}K_{DW}([\hat{Q}_x(\rho_{ABE})]^{\otimes n}),
\end{equation}
where $K_{DW}(\rho_{XBE}):=I(X:B|\hat{T})_\rho - I(X:E|\hat{T})_\rho$, and the maximum in~\eqref{eq_key_iid} is taken over POVMs of the form $\{\hat{Q}_x^{\otimes n}\}_{x\in \mathcal{X}}$, and  channels $R$, such that $\hat{T}^{\times n}=R(X)$.
\end{definition}

We have introduced the $K^{iid}$, as it is easier to study its behavior than that of $K_\rightarrow$. While, as we show further, $K^{iid}$ is to some extent non-lockable, $K_\rightarrow$ still can be lockable. We have finally $K^{iid}(\rho_{ABE})\leq K_{\rightarrow}(\rho_{ABE}) \leq C_D(\rho_{ABE})$, for any tripartite quantum state $\rho_{ABE}$. 

\subsection{Key distillable by LOCC operations}
\label{def_K_D}
Distillable key $K_D$ between Alice and Bob from a quantum bipartite state $\rho$ by means of  two-way  LOCC operations is given as \cite{pptkey,keyhuge}:
\begin{equation}
 K_D(\rho):=\inf_{\epsilon >0}\limsup_{n\rightarrow \infty}\sup_{\Delta \in LOCC}\left\{\frac{\log_2 d}{n} \ | \  \Delta(\rho^{\otimes n })\approx_{\epsilon} \gamma_d\right\},
\end{equation}
where $\gamma_d$ is  a $d$-dimensional private state with $\log_2 d$ secure bits, and $\approx_{\epsilon}$ denotes $\epsilon$-closeness in the trace norm $||\cdot||_1$.

$K_D$ quantifies the amount of key secure against a quantum adversary who holds a purification of the state $\rho_{AB}$ can be obtained from asymptotically many copies of this state, in the form of a private state. 
Importantly, it can be shown~\cite{keyhuge,karol-PhD} that for a pure tripartite state $\psi_{ABE}$ with corresponding state $\rho_{AB}=\tr_E \psi_{ABE}$ one has
\begin{equation}
\label{eq:CK}
C_D(\psi_{ABE})=K_D(\rho_{AB}).
\end{equation}
 Therefore, in the worst case, that is when the adversary Eve holds a purifying system of $\rho_{AB}$, considering distillation of private states by LOCC operations or the ideal key states $\tau$ by LOPC operations yields the same rate. This allows us to interchange the use of $C_D$ and $K_D$ if needed.

\section{Bound on the leakage of private randomness}
In this section, we focus on distributed scenario of private randomness distillation~\cite{YHW}. In this scenario, two honest parties share $n$ copies of a bipartite state $\rho_{AB}$. They use local unitary operations and dephasing channel to produce independent randomness private against Eve, who holds the purifying system and the environment of the dephasing channel.  Depending on whether free or no local noise (in the form of a maximally mixed state) and free or no communication are allowed, we have four different settings for the distributed private randomness distillation. Theorem 2 in~\cite{YHW} shows the achievable rate regions (of private randomness distillable locally for each of the parties). For convenience and self-consistency of the paper, we restate it in the following. Here $R_G(\rho_{AB}) := \log_2|A|+\log_2|B|-S(\rho_{AB})$ stands for {\it global purity},
while $R_A$ is private randomness localizable by party $A$ in respective scenario (similarly for $B$). 
{\theorem \label{two-side} 
The achievable rate regions of $\rho_{AB}$ are:
\begin{enumerate}
\item for no communication and no noise, 
$  R_A       \leq \log_2|A| - S(A|B)_+ $, 
$  R_B       \leq \log_2|B| - S(B|A)_+ $, and
 $ R_A + R_B \leq R_G$,
where $[t]_+=\max\{0,t\}$;

\item for free noise but no communication, 
$R_A       \leq \log_2|A| - S(A|B)$, 
$R_B       \leq \log_2|B| - S(B|A)$, and 
$R_A + R_B \leq R_G$;

\item for free noise and free communication, 
$R_A\leq R_G$, $R_B\leq R_G$, and $R_A+R_B\leq R_G$;

\item for free communication but no noise, 
$R_A       \leq \log_2|AB| - \max\{S(B),S(AB)\}$, 
$R_B       \leq \log_2|AB| - \max\{S(A),S(AB)\}$, and
$R_A + R_B \leq R_G$. 
\end{enumerate}
Further the rate regions in settings 1), 2), 3) are tight.}

\medskip
We consider then a local leakage at Alice's side, by a side channel consisting of local unitary $U_{A\to A'a} $ transformation of a system $A$ into a system $A'a$, followed by partial trace operation on system $a$, which implies leakage of this system to Eve. Before we get the proposition, we need an auxiliary technical fact. 

\begin{fact}
\label{fact1}
For any two numbers $x$ and $y$, 
\ben
\max\{0,x\}-\max\{0,y\}\le |x-y|.
\een
\end{fact} 
This can be checked directly by considering the two cases of $y\le 0$ and $y>0$. Now we are in position to formulate and prove the main result for this section.

{\proposition For a bipartite state $\rho_{AB}$ subjected to a side channel $\tr_a \circ U_{A\to A'a}$, there is
\ben
R_A(\rho_{AB})-R_A(\rho_{A'B}) \le \log_2 |a|+S(a),
\een
in the four settings presented in Theorem~\ref{two-side}.
}

\medskip
\begin{proof}
Setting 1) is reduced to setting 2) by noticing auxiliary Fact~\ref{fact1}. Then we have 
\ben
&&R_A(\rho_{AB})-R_A(\rho_{A'B})\\
&&=\log_2|A|-\max\{0,S(A|B)\}-\nonumber\\
&&\quad [\log_2|A'|-\max\{0,S(A'|B)\}]\\
&&=\log_2|a|+\max\{0,S(A'|B)\}-\max\{0,S(A|B)\}\\
&&\le\log_2|a|+|S(A'|B)-S(A|B)|\\
&&=\log_2|a|+|S(A'B)-S(A'aB)|\\
&&\le \log_2|a|+S(a),
\een
where the first inequality comes from the auxiliary Fact~\ref{fact1} and the last inequality from the subadditivity of entropy~\cite{RevModPhys.50.221}.

The proof for setting 3) is straightforward.
\ben
&&R_A(\rho_{AB})-R_A(\rho_{A'B})\\
&&=\log_2|AB|-S(AB)-[\log|A'B|-S(A'B)]\\
&&=\log_2|a|+[S(A'B)-S(A'aB)]\\
&&\le \log_2|a|+S(a).
\een

The proof for setting 4) can be reduced to setting 1) by noticing 
\ben
&&R_A(\rho_{AB})-R_A(\rho_{A'B})\\
&&=\log_2|AB|-\max\{S(B),S(AB)\}\nonumber\\
&&-[\log_2|A'B|-\max\{S(B),S(A'B)\}]\\
&&=\log_2|a|+\max\{0,S(A'|B)\}-\max\{0,S(A|B)\}\\
&&\le \log_2|a|+S(a).
\een
\end{proof}

\subsection{Distillable key of maximally correlated states is strongly non-lockable}
In the following theorem we show that distillable key of a MCS is strongly non-lockable. A pure bipartite state is a special MCS in its Schmidt basis.
\begin{theorem}
\label{thm:pure}
 For a maximally correlated state $\rho_{AB}$ defined through expression~\eqref{MCS}, after leakage of system $a$ from Alice to Eve the distillable key $K_D$ decreases by no more than S(a).
\end{theorem} 

\begin{proof}
For a MCS $\rho_{AB}$, we have that $K_D(\rho_{AB})=E_D(\rho_{AB})=E_r(\rho_{AB})=S(B)-S(AB)$. Suppose an isometry $U: A\to A'a$, and after the leakage of subsystem $a$ to Eve, then the shared state between Alice and Bob is $\rho_{A'B}$. By Devetak-Winter protocol, we have $K_D(\rho_{A'B})\ge S(B)-S(A'B)$ (this is the other direction of DW protocol). Therefore the loss of the distillable key can be upper bounded as follows,
\ben
&K_D&(\rho_{AB})-K_D(\rho_{A'B})\\
&\le& S(B)-S(AB)-[S(B)-S(A'B)],\\
&=&S(A'B)-S(A'aB),\\
&\le&S(a),
\een
where we use $S(A'aB)=S(AB)$ since $U$ is an isometry, and sub-additivity of the von Neumann entropy. 
\end{proof}

\begin{corollary}
    The BB84 protocol \cite{BB84}, realized by means of the CSS codes, has a non-lockable rate.
\end{corollary}  

\begin{proof}
In \cite{shor-preskill} it is shown, that 
such a protocol, if applied coherently, is equivalent to distillation of maximally entangled states. Hence, if the prepare-measure version of BB84 was lockable, i.e. the key upon tracing out
some system $a$ would drop down by more than $S(a)$, so would be the drop of it
for the coherent version. The latter is 
however forbidden by the Theorem \ref{thm:pure}. 
\end{proof}
In the next Section we generalise Theorem~\ref{thm:pure} to Schmidt-twisted pure states $\widetilde{\gamma}_{ABA'B'}$ introduced in equations~\eqref{Schmidt_twist} and~\eqref{Schmidt_pure}.

\section{Lower bound for the drop of generated key under leakage of a system}
\label{sec:lb_for_generated_key}
In this section, we investigate how much the generated key drops after leakage of a system. We start from subsection~\ref{subA} where we prove  how much is the key rate drops for an irreducible private state when the system leaks from the shield part of Alice to Eve. Next, in subsection~\ref{subB} we generalize  the proof technique to all states and different types of leakage, such as erasure of a system or copying of a system. In turn, we prove the main result contained in Theorem~\ref{thm:raw_key}, saying that the raw key of a one-way Devetak-Winter protocol is non-lockable. In subsection~\ref{subC} by exploiting the concept of smooth min- and max- entropy, we show that the single-shot key rate is non-lockable. Finally, in subsection~\ref{subD} we derive a lower bound on the loss of the two-way distillable key  for the irreducible Schmidt-twisted pure states.

\subsection{Bound on the key drop by leakage from a irreducible private state}
\label{subA}
In this subsection, we provide a simple lower bound on the distillable key in the presence of leakage of subsystem $a$ from irreducible private states defined in Section~\ref{sec:notation} from the shield part, as well as from Alice's side in general.  In all cases, we show that the key drops by no more than $2S(a)$. We start our considerations from  the case of the  leakage from the shield part:
\begin{observation}
    For an irreducible private state $\gamma_{AA'BB'}$, with $A'=aA''$, there is
    \begin{equation}
    K_D(\gamma_{AA''BB'}) \geq K_D(\gamma_{AA'BB'}) -2 S(a).
    \end{equation}
    \label{obs:shield}
\end{observation}
\begin{proof}
The distillable key of an irreducible private state $\gamma_{AA'BB'}$ reads $\log_2 d_k$.
Let us then divide system $A'$ into ${\hat A}a$.  The Devetak--Winter protocol applied to the key part (from Bob to Alice) reads: 
\begin{equation}
I(A:B)_\rho - I(B:Ea)_\rho = \log_2 d_k - I(B:E) - I(B:a|E),
\end{equation}
where $I(B:a|E)= S(BE)+S(aE)-S(E)-S(BaE)$ is the conditional mutual information, which
follows from the chain rule. From $I(X:Y)\leq 2\min\{ S(X),S(Y)\}$ and
the chain rule, we conclude that $I(B:a|E)\leq 2 \min\{S(a),S(B),S(aE),S(BE)\}\leq 2S(a)\leq 2 \log_2 |a|$ \cite{Shirokov_2017}.
This, due to $I(B:E)=0$, as the state is the private state, proves our observation.
\end{proof}
Now, we will extend the statement of Observation~\ref{obs:shield} to the leakage  from the irreducible private state in a general way, not necessarily from its shield part. To do so, let us first prove the following technical lemma:

\begin{lemma}
For a cqq state $\rho_{XAaE}=\sum p_i\ketbra{i}{i}_X\otimes \rho^i_{AaE}$, after the leakage of system $a$ from Alice to Eve, the following holds 
\begin{equation}
[I(X:Aa)-I(X:E)]-[I(X:A)-I(X:aE)]\leq 2S(a).
\end{equation} 
\label{lem:Dong}
\end{lemma}

\begin{proof}
The proof goes by straightforward calculations and strong subadditivity. 
\begin{eqnarray}
\label{ddd}
&&[I(X:Aa)-I(X:E)]-[I(X:A)-I(X:aE)]\\
&&=I(X:a|A)+I(X:a|E)\\
&&=S(a|A)-S(a|AX)+S(a|E)-S(a|EX)\\
&&=S(a|A)+S(a|E)-\sum_ip_i[S(a|E)_i+S(a|A)_i]\\
&&\leq 2S(a),
\end{eqnarray}
where the inequality comes from the facts that $S(a|A)\le S(a)$, $S(a|E)\le S(a)$, and $S(a|E)_i+S(a|A)_i\geq 0$ for each index $i$ which follows from the strong subadditivity. Namely, considering purification of $\rho_{aAE}$ to $|\psi_{\rho}\>_{aAEE'}$ we can write $S(a|EE')+S(a|A)=0$, since $S(aEE')=S(A)$, and $S(EE')=S(aA)$. But using strong subadditivity we write $S(a|EE')\leq S(a|E)$, so $S(a|E)+S(a|A)\geq S(a|EE')+S(a|A)=0$. This argumentation holds for every index $i$ in expression~\eqref{ddd}.
\end{proof}

\begin{proposition}
For an irreducible private state $\gamma_{AA'BB'}$, with $AA'=a\tilde{A}$, after the leakage of system $a$ from Alice to Eve, there is
\begin{equation}
K_D(\gamma_{\tilde{A}BB'})\ge K_D(\gamma_{AA'BB'})-2S(a).
\end{equation}
 \label{obs:shield_gen}
\end{proposition}
\begin{proof}
Denote $\gamma_{AA'BB'E}$ as the purification of $\gamma_{AA'BB'}$ when Eve's system $E$ is included. 
Consider then this state measured on $B$ in computational
basis, producing a random variable $X$. 
Further notice that we have the following chain of (in)equalities:
\begin{eqnarray}
K_D(\gamma_{AA'BB'})&=&I(X:A)-I(X:E)\\
&\leq& I(X:AA')-I(X:E)\\
&=&I(X:\tilde{A}a)-I(X:E)\\
&\leq& I(X:\widetilde{A})-I(X:aE)+2S(a)\\
&\leq& K_D(\gamma_{\widetilde{A}BB'})+2S(a)
\end{eqnarray}
The first equality is due to the fact that $\gamma$ is irreducible, hence $K_D(\gamma_{ABA'B'})=\log_2 d_k = I(X:A) = I(X:A)- I(X:E)$, as $I(X:E)=0$ due to privacy from Eve of the system $B$ under measurement. The first inequality
is due to data processing inequality \cite{Nielsen-Chuang} implying $I(X:A)\leq I(X:AA')$. We next observe that
the unitary transformation does not change the mutual information, hence $I(X:AA')= I(X:{\tilde A}a)$.
Finally we note that 
\begin{align}
    I(X:{\tilde A}a) - I(X:E) - [I(X:{\tilde A}) - I(X:aE)] \leq 2S(a),
\end{align}
where the inequality is due to Lemma \ref{lem:Dong} by identifying $A$ with ${\tilde A}$. This finishes the proof.
\end{proof}
The upper bounds on the key in Observation~\ref{obs:shield} and Proposition~\ref{obs:shield_gen} are tight, which implies that $K_D$ is not strongly non-lockable in general. The example comes from a variant of the superdense coding protocol.

{\example Consider the private state $\gamma_{AA'B}$ where $B'$ is a trivially 1-dimensional subsystem and the purification of the state with Eve's system $E$ is of the form 
\begin{equation}
\frac{1}{\sqrt{4}}\sum_{i=0}^{3}\ket{ii}_{AB}\otimes (\sigma_{A'}^{i}\otimes I_E)\ket{\Phi}_{A'E},
\end{equation}
where $\sigma_{A'}^{i}$ are the Pauli unitary operators acting on the subsystem $A'$ and $\ket{\Phi}_{A'E}=\frac{1}{\sqrt{2}}(\ket{00}_{A'E}+\ket{11}_{A'E})$.
A simple observation is that $K_D(\gamma_{AA'B})=2$ and after the leakage of the shielding qubit $A'$ to Eve, $K_D=0$. 
}

The same holds if the leakage takes place on system $B'$,
unless it is given to Eve. Hence, given that the leakage happens only on the shielding system of an irreducible private state, the key drops down by at most twice the entropy of the system, and in some cases, it can equal to 2.

\subsection{The raw key of a one-way Devetak-Winter protocol is non-lockable}
\label{subB}
We now generalize the result from subsection~\ref{subA} to
all states that are the output of key-generation protocol. In practice, they differ from private states considered above.
This is because the process of key generation is usually
not coherent.
In that
we also narrow to one-way key distillation.
We will first need the following observation:
\begin{observation}
For a cq state $\rho_{x(XY)}$,
\begin{equation}
    I(x:Y|X) \leq H(x).
\end{equation}
    \label{obs:classical}
\end{observation}
\begin{proof}
It is convenient to rewrite $I(x:Y|X)$
as
\begin{equation}
    I(x:Y|X) = S(x|X) - S(x|YX).
\end{equation}
The state $\rho_{x(XY)}$ is separable in cut
$x:(XY)$, hence $S(x|YX) \geq 0$ \cite{RMPK-quant-ent}. We can thus neglect this term, obtaining an upper bound
\begin{equation}
    I(x:Y|X) \leq S(x|X).
\end{equation}
Since $\rho_{x(X)}$ is also a cq state, we can further expand $S(x|X)$ as
\begin{eqnarray}
    S(x|X)= H(x) + \sum_x p(x)  S(\rho_{X|x}) \nonumber\\- S\left(\sum_x p(x)\rho_{X|x}\right) \leq H(x),
\end{eqnarray}
where the last inequlity is due to
concavity of the von Neumann entropy.
\end{proof}
\begin{lemma}
For a state $\rho_{aABET}$, there is
\begin{eqnarray}
\label{relation}
    I(A:B|T) - I(A:Ea|T) \geq \nonumber\\ I(Aa:B|T)-I(aA:E|T) -c S(a),
\end{eqnarray}
with $c=2$.
Moreover, when state $\rho_{a(ABET)}$ is a
cq state, then the bound holds for $c=1$.
\label{lem:bound}
\end{lemma}
\begin{proof}
The first part of the lemma is obtained by direct calculations. Namely, we have the following:
\begin{eqnarray}
&I&(Aa:B|T)-I(Aa:E|T)-\left[I(A:B|T)-I(A:Ea|T)\right]=\nonumber\\
&-&S(a|ABT)+S(a|ET)\leq 2S(a).
\end{eqnarray}
To show the second part of the statement, when we deal with a cq state, it is enough to notice that $S(a|ABT)\geq 0$.
%We first prove the case of $c=6$.
%We have the following chain of (in)equalities explained below:
%\begin{eqnarray}
%I(A:B|T) &-& I(A:Ea|T) \pm I(a:B|AT) =\nonumber\\
%I(aA:B|T) &-& I(a:B|AT) -I(A:Ea|T) \geq \\
%I(aA:B|T) &-& 2S(a) - I(A:Ea|T) =\\
%I(aA:B|T) &-& 2S(a)- \nonumber \\
%(I(A:E|T) &+& I(a:A|ET)) \geq\\
%I(Aa:B|T) &-& I(A:E|T) -4S(a) = \\
%I(Aa:B|T) &-& I(Aa:E|T) \nonumber\\- I(a:E|AT) &-& 4 S(a)\geq\\
%I(Aa:B|T) &-&I(Aa:E|T) -6 S(a).
%\end{eqnarray}
%In the first inequality we use chain rule,
%and further bound on term $I(a:B|AT)$. In the second equality we use chain rule,
%and further bound the term $I(a:A|ET)$.
%Last equality also stems from the chain rule, finalised again by lower bound of $I(a:E|AT)$ by $2S(a)$.  To see the proof
%for $c=3$, one needs to note that if the state $\rho_{a(TABE)}$ is cq, we can use
%Observation \ref{obs:classical}. Hence, in this case there is $I(a:B|AT)\leq H(a)$ and
%the same bound holds for $I(a:A|ET)$ and $I(a:E|AT)$. 
\end{proof}
We have considered above a drop of a system on the side of a sender of one-way communication during key distillation via Devetak--Winter protocol \cite{DevetakWinter-hash}. We now show that
similar result holds for the party who, in their protocol, receives only the data.
\begin{corollary}
For a state $\rho_{ABbET}$, there is
\begin{equation}
\label{onBob}
\begin{split}
&I(A:B|T)-I(A:Eb|T)\geq \\
&I(A:Bb|T)-I(A:E|T)-cS(b),
\end{split}
\end{equation}
with $c=4$. Moreover, if the state $\rho_{b(ABET)}$ is a cq state, then the bound holds for $c=2$. 
\end{corollary}

\begin{proof}
The proof  follows from the following chain of inequalities:
 \begin{eqnarray}
%\begin{split}
&I&(A:B|T)-I(A:Eb|T)\pm I(b:A|BT)=\\
&I&(A:Bb|T)-I(A:Eb|T)-I(b:A|BT)\geq\\
&I&(A:Bb|T)-I(A:Eb|T)-2S(b)=\\
&I&(A:Bb|T)-I(A:E|T)-I(b:A|ET)-2S(b)\\
&\geq&I(A:Bb|T)-I(A:E|T)-4S(b).
%\end{split}
 \end{eqnarray}
 We first focus on the case $c=4$.
The first equality comes from the chain rule, while the first inequality from bound on $I(b:A|BT)$. Similarly, the second equality follows from the chain rule, and following inequality from bounding the term  $I(b:A|ET)$. 
Regarding the case $c=2$ we note
that when system $b$ is classical,
then both terms $I(b:A|ET)$ and $I(b:A|BT)$ are bounded by $S(b)$ by Observation \ref{obs:classical}, which proves the thesis.
\end{proof}

Owing to the fact that the raw key is 
classical, it is also realistic to 
assume that the leakage will be through copying rather than the theft of data. We therefore
consider this case below.
\begin{corollary}
For a state $\rho_{AaBET}$, there is
\begin{equation}
\label{copying}
\begin{split}
&I(Aa:B|T)-I(Aa:Ea|T)\geq\\
&I(Aa:B|T)-I(Aa:E|T)+cS(a),
\end{split}
\end{equation}
with $c=2$.
\end{corollary}

\begin{proof}
To prove expression~\eqref{copying}, we write the following chain of inequalities:
\begin{eqnarray}
%\begin{split}
&I&(Aa:B|T)-I(Aa:Ea|T)\pm I(Aa:E|T)=\\
&I&(Aa:B|T)-I(Aa:E|T)-I(a:Aa|ET)\geq\\
&I&(Aa:B|T)-I(aA:E|T)-2S(a).
%\end{split}   
\end{eqnarray}
The first equality follows from the chain rule, and further we bound the term $I(a:Aa|ET)$.
\end{proof}

To conclude about the non-lockability of the 
raw key obtained in the one-way protocol
we  base on  the main result of Devetak and Winter in \cite{DevetakWinter-hash}, invoked in Sec. \ref{sec:LOPC}.

Let ${\cal P}$ be a part of the protocol of one-way key distillation after Alice have performed measurement $Q_x$, i.e.  after producing  a state of the form $Q_x(\rho_{ABE}^{\otimes n})=\omega_{TABE}^{(n)} =
\sum_{t,x} R(t|x)P(x) |t\>\<t| \otimes |x\>\<x|\otimes Tr_A(\rho_{ABE}(Q_x)\otimes {\mathbf{1}}_{BE})$.  It consists of an error correction and a privacy amplification operations applied to the state $\omega_{TXBE}$ \cite{RennerPhD} and ${\cal P}$ is the part of total protocol, which generates the key from the {\it raw key} at Alice's side. Let also the rate of ${\cal P}$ be denoted as $\kappa$. 
In the above theorem, the state of the raw key is represented by $\omega_{TXBE}^{(n)}$. We assume also
that system of Alice is represented by $A\equiv Xx$,
where $x$ will be given to Eve in the process of leakage.
We have then an immediate result:
\begin{theorem}
    The raw key of a one-way Devetak-Winter protocol is non-lockable: for any state $\omega_{T(Xx)BE}^{(n)}$ generated by measurement $Q_x$ on $n$ copies of $\rho_{AaBE}$, and for any random variable $T=R(X)$, there is 
    $\kappa({\cal P}(\omega_{T(Xx)BE}^{(n)}))\geq \kappa({\cal P}(\omega_{TXB(Ex)}^{(n)})) - H(x)/n$.
    \label{thm:raw_key}
\end{theorem}

\begin{proof}
Let us denote the states where the raw key is presented, in both cases, when the system $x$ is with Alice and Eve by $\omega^{(n)}_{TXxBE}$ and $\omega^{(n)}_{TXB(Ex)}$, respectively. Denoting by ${\cal P}$ the one-way key distillation protocol applied to both states, we evaluate its rates $\kappa$ as:
\begin{align}
&\kappa\left({\cal P}(\omega^{(n)}_{TXxBE})\right)=\frac{1}{n}\left[I(Xx:B|T)-I(Xx:E|T)\right],\\
&\kappa\left({\cal P}(\omega^{(n)}_{TXB(Ex)})\right)=\frac{1}{n}\left[I(X:B|T)-I(X:Ex|T)\right].
\end{align}
Applying the statement 
from Lemma~\ref{lem:bound}, and using the fact that $x$ is classically correlated with the rest of the systems, we can write
\begin{align}
\kappa\left({\cal P}(\omega^{(n)}_{TXB(Ex)})\right)\geq \kappa\left( {\cal P}(\omega^{(n)}_{TXxBE})\right)-\frac{H(x)}{n}.
\end{align}
Hence, whenever entropy $H(x)$ scales linearly with number of copies $n$, i.e. when $H(x)=\alpha n$, where $\alpha$ is a constant, the raw key drops by constant factor. However, when the dependence is sublinear in $n$, the resulting raw key does suffer from the leakage.
\end{proof}
The same statement as in Theorem~\ref{thm:raw_key} can be made in the case of system leakage $b$ from Bob to Eve, or of copying the system $a$ from Alice to Eve. Denoting by $(\omega^{(n)}_{TX(Bb)E},\omega^{(n)}_{TXB(Eb)})$ and $(\omega^{(n)}_{TXxBE},\omega^{(n)}_{TXxB(Ex)})$ the pairs of states containing the raw key in the case of leakage of Bob's system and of copying, respectively, we formulate the following:
\begin{observation}
The raw key of a one-way Devetak--Winter protocol is non-lockable in the case of system leakage from Bob to Eve and of copying a system from Alice to Eve. In particular, the raw key rates before and after the process of leakage (copying) satisfy, respectively:
\begin{align}
&\kappa({\cal P}(\omega_{TX(Bb)E}^{(n)}))\geq \kappa({\cal P}(\omega_{TXB(Eb)}^{(n)})) - 4S(b)/n,\\
&\kappa({\cal P}(\omega_{T(Xx)BE}^{(n)}))\geq \kappa({\cal P}(\omega_{TXxB(Ex)}^{(n)})) - 2S(x)/n.
\end{align}
Whenever entropies $S(x)$ and $S(b)$ scale linearly or sublinearly with $n$ the raw key drops down by a constant factor or does not change in the limit of large $n$.
\label{obs:2}
\end{observation}

\subsection{Single-shot key rate approach after leakage system to Eve}
\label{subC}
By the result of~\cite{RenesRenner2012}, one can deduce how much smooth min entropy $H^{\epsilon}_{\operatorname{min}}$ drops after the leakage of system $x$ to Eve (see subsection~\ref{subMinMax} for definitions).
\begin{lemma}[adaptation of Lemma 5 from~\cite{RenesRenner2012}]
The smooth min entropy $H^{\epsilon}_{\operatorname{min}}$ is non-lockable. It means that after leakage of a system $x$ to Eve, the following inequality holds:
\begin{equation}
\label{Hmin_bound}
   H^{\epsilon}_{\operatorname{min}}(Xx|E)\leq H^{\epsilon}_{\operatorname{min}}(X|Ex)+\operatorname{log}_2|x|,
\end{equation} 
where $|x|$ denotes dimension of the system $x$. 
\end{lemma}
Using the above result, one can show that the single-shot key rate is non-lockable. Namely, before and after leakage of a system $x$ to Eve, the key rates are respectively:
\begin{equation}
\label{keys}
\begin{split}
&K^{(1)}(\rho_{(Xx)BE})=H_{\min}^{\epsilon}(xX|E)-H_{\max}^{\epsilon}(xX|B),\\
&\widetilde{K}^{(1)}(\rho_{XB(Ex)})=H_{\min}^{\epsilon}(X|Ex)-H_{\max}^{\epsilon}(X|B).
\end{split}
\end{equation}
Applying data processing theorem~\cite{RenesRenner2012} to the expression of~\eqref{keys} we have that $H_{\max}^{\epsilon}(xX|B)\leq H_{\max}^{\epsilon}(X|B)$. Thanks to this, we conclude that the key drops by no more than $\operatorname{log}_2|x|$.

Finally, by observing that, in the limit $n\rightarrow \infty$, the min- and max- entropies converge to the conditioned Shannon entropy \eqref{minmaxconv}, we can conclude that  the right hand side of~\eqref{Hmin_bound} gives $n\operatorname{log}_2|x|$. Whenever system $x$ is of $n$ qubits, and $S(x)>{1\over 4}n$ holds, this bound is smaller than the bound $4S(x)$ discussed in the Observation \ref{obs:2}.

\subsection{Lower bound on the loss of the distillable key for the irreducible Schmidt-twisted pure states}
\label{subD}
The bounds shown in the 
previous sections are independent of the correlations of the erased system $a$ with the rest of the system. However, it is 
intuitive that the less $a$ is correlated
the smallest should be drop of the key upon loss of $a$. This motivates us to search for
a bound which is dependent on these correlations. 
 
To show that the key sometimes does not leak too fast, we propose a particular strategy to be taken after erasure of subsystem of the state.
It is based on the so called {\it fidelity of recovery} \cite{Fawzi2015,Seshadreesan2015}.

As we will see, this approach will
lead us to a bound on a two-way distillable key for private states.
Namely, after the loss of a subsystem $a$ 
of a system $Aa$, Alice is applying the best map $\Gamma_{A\rightarrow A\widetilde{a}}$ that recovers $a$ in some form $\tilde{a}$.
She then applies the same one-way protocol on system $A{\tilde a}$. Denoting by $F(\rho_{AaBE},\widetilde{\rho}_{A\widetilde{a}BE})$ the Uhlmann fidelity between quantum states~\cite{Uhlmann}, the fidelity of recovery
reads
\begin{equation}
\label{F_recovery}
F_R(a;BE|A):=\sup_{\Gamma_{A\rightarrow A\widetilde{a}}}F\left(\rho_{AaBE},\Gamma_{A\rightarrow A\widetilde{a}}(\rho_{ABE})\right),
\end{equation}
where $\rho_{AaBE}$ with $\rho_{ABE} = \tr_a \rho_{AaBE}$, and we suppressed identity $\mathbf{1}_{BE}$ in the action of recovery map $\Gamma_{A\rightarrow A\widetilde{a}}(\rho_{ABE})\equiv (\mathbf{1}_{BE}\otimes \Gamma_{A\rightarrow A\widetilde{a}}) (\rho_{ABE})=\widetilde{\rho}_{A\widetilde{a}BE}$. We will call $\rho_{A\widetilde{a}BE}$ a recovered state. 
It is proven that there is an appealing lower bound on the formula \eqref{F_recovery} in terms
of the conditional mutual information \cite{Fawzi2015,Seshadreesan2015}:
\begin{equation}
F_R(\rho_{AaBE}) \geq 2^{-I(a:BE|A)}.
\end{equation}
This allows us for estimating closeness of single copy one-way secure key $K^{(1)}_{\rightarrow}$ between state $\rho_{AaBE}$ and its recovered version $\widetilde{\rho}_{A\widetilde{a}BE}$. In what follows, we use lemma V.3 of \cite{Nowakowski2016} for the case of triparite states (for biparite states it needs correction, see Lemma \ref{lem:biparite} presented in the Appendix). 

\begin{observation} 
For any state $\rho_{AaBE}$ and its recovered version $\widetilde{\rho}_{A\widetilde{a}BE}=\Gamma_{A\rightarrow A\widetilde{a}} (\rho_{ABE})$, where $\rho_{ABE} = \tr_a \rho_{AaBE}$, and $\Gamma_{A\rightarrow A\widetilde{a}}$ is recovery map, the following relation holds:
\label{obs_Now}
\begin{equation}
|K_{\rightarrow}^{(1)}(\widetilde{\rho}_{A\widetilde{a}BE})-K^{(1)}_{\rightarrow}(\rho_{AaBE})| \leq 8 \delta \log_2 d_{Aa} + 4h(\delta),
\label{eq:key_one_close}
\end{equation}
with $\delta = \sqrt{1-2^{-I(a:BE|A)}}$ and $h(\cdot)$ denoting the binary Shannon entropy.
\end{observation}
%It is easy to show the following bound:
This observation follows directly from the Fuchs--van de Graaf
inequality \cite{Fuchs:vandeGraaf:1999}, which for two arbitrary states $\rho,
\sigma$ reads ${1\over 2}||_1\rho -\sigma||\leq \sqrt{1-F(\rho,\sigma)}$, and the fact that (not regularised) one-way distillable key is asymptotically
continuous (see Lemma V.3  in \cite{Nowakowski2016}). We use fidelity which is calculated for a map $\Gamma_{A\rightarrow A\widetilde{a}}$ maximising the fidelity of recovery in~\eqref{F_recovery}.

The above considerations hold for one copy of the state $\rho_{A\widetilde{a}BE}$. Now we shall discuss and find an upper bound for the regularised version, $K_{\rightarrow} = \lim_n {1\over n} K_{\rightarrow}^{(1)}(\rho^{\otimes n})$. The above reasoning cannot be applied straightforwardly to this case, because the closeness of $\rho$ and $\sigma$ in trace norm $\frac{1}{2}||\rho-\sigma||_1$ does not imply their closeness after taking many copies, when one considers  $\frac{1}{2}||\rho^{\otimes n}-\sigma^{\otimes n}||_1$.

Nevertheless, we can extend the above result to a class of {\it Schmid-twisted irreducible private states}.  Let us recall that   the one-way i.i.d. version of a secure key, $K^{iid}$,  is the key distillable by one-way communication via first measuring and post-processing it in an i.i.d. way on Alice's side. That is, the measurement $Q_x$ on Alice's side is performed 
identically and independently on each copy
of the state, rather than globally, and further classical information comes from a variable $\hat{T}$ that is identical on each copy (see Sec. \ref{sec:LOPC} for a full definition).

We can now prove the result inspired by Observation~\ref{obs_Now} in the case of $K^{iid}$:  
\begin{theorem}
\label{thm:delta_bound} 
Let $K^{iid}$ be one-way i.i.d. version of secure key, as in Definition~\ref{key_iid}. Denoting the original state by $\rho_{aABE}$,  the following inequality holds:
\begin{equation}
\label{chain_thm1}
\begin{split}
K^{iid}(\rho_{ABE})\geq  K^{iid}(\rho_{aABE})-(4\delta \log_2(d_ad_Ad_B^2)+4h(\delta)),
\end{split}
\end{equation}
where $\delta = \sqrt{1-2^{-I(a:BE|A)}}$, $I(a:BE|A)$ is a conditional mutual information calculated on respective systems, and $h(\cdot)$ denotes the binary Shannon entropy.
\end{theorem}

\begin{proof}
Let $\hat{Q}_x^*$ be the optimal measurement
realizing $K^{iid}(\rho_{AaBE})$, where
$\rho_{AaBE}=|\psi_{AaBE}\>\<\psi_{AaBE}|$, and let $\widetilde{\rho}_{\widetilde{a}ABE}$ be the state after application of the recovery map $\Gamma_{A\rightarrow A\widetilde{a}}$ to the state $\rho_{ABE}=\tr_a\rho_{aABE}$. As we have argued below Observation~\ref{obs_Now}, there is
$||\widetilde{\rho}_{\widetilde{a}ABE}-\rho_{AaBE}||_1\leq \delta$, and the same holds for this pair of states after application of the measurement $\hat{Q}_x^*$, so $||\widetilde{\rho}'_{\widetilde{a}ABE}-\rho'_{XABE}||_1\leq \delta$, where $\rho'_{XABE}=\hat{Q}_x^*(\rho_{aABE})$, with $X$ denoting the outcome of the measurement.
Now, applying Definition~\ref{key_iid} to our case, one has
\begin{eqnarray}
%\begin{split}
&&K^{iid}(\rho_{aABE})= \lim_{n\rightarrow \infty} \frac{1}{n} \max_{\hat{T}|X}K_{DW}([Q_x^*(\rho_{aABE})]^{\otimes n})\\
&&=\lim_{n\rightarrow \infty} \frac{1}{n} \max_{\hat{T}|X}\left(I(X:B|\hat{T})_{\rho^{'\otimes n}}-I(X:E|\hat{T})_{\rho^{'\otimes n}}\right)\\
&&=\max_{\hat{T}|X} \left(I(X:B|\hat{T})_{\rho'}-I(X:E|\hat{T})_{\rho'}\right)\\
&&=I(X:B|\hat{T}^*)_{\rho'}-I(X:E|\hat{T}^*)_{\rho'}.
%\end{split}
\end{eqnarray} 
To obtain the second line, we use $K_{DW}(\rho_{XBE})=I(X:B|\hat{T})_\rho - I(X:E|\hat{T}
)_\rho$. To obtain the third line, we exploit the additivity of the conditional mutual information. To get the last line, we introduce the quantity $\hat{T}^*$ attaining the maximum value of $\hat{T}$. 
On the other hand, by similar lines,
there is 
\begin{equation}
    K^{iid}(\widetilde{\rho}_{\widetilde{a}ABE}) \geq 
    I(X:B|\hat{T}^*)_{\widetilde{\rho}'} -I(X:E|\hat{T}^*)_{\widetilde{\rho}'},
    \end{equation}
where $\hat{T}^*$ is the value of $\hat{T}$ that attains maximum
in the formula for $K^{iid}(\rho_{aABE})$, and
$X$ is the outcome of measurement $\hat{Q}^*_x$ on $\widetilde{\rho}_{\widetilde{a}ABE}$.
Finally, to prove expression~\eqref{chain_thm1} we write the following chain of inequalities:
\begin{eqnarray}
\label{close0}
&&K^{iid}(\rho_{ABE}) \geq K^{iid}(\tilde{\rho}_{\tilde{a}ABE})\\
&&\geq I(X:B|\hat{T}^*)_{\widetilde{\rho}'}-I(X:E|\hat{T}^*)_{\widetilde{\rho}'}\\
&&\geq I(X:B|\hat{T}^*)_{\rho'}-I(X:E|\hat{T}^*)_{\rho'}-(4\delta \log_2 d_B +\nonumber\\
&&+4\delta \log_2(d_ad_Ad_B)+4h(\epsilon))\\
&&\geq K^{iid}(\rho_{aABE})-(8\delta \log_2 d_B+4\delta \log_2 (d_ad_A)+\nonumber\\
&&+4h(\delta))\nonumber\\
&&=K^{iid}(\rho_{aABE})-(4\delta \log_2(d_ad_Ad_B^2)+4h(\delta)).
\label{eq:close_final}
\end{eqnarray}
The first inequality follows from the fact that  the operation of recovery, since it is local, does not increase the amount of the key. To obtain the third line we use the fact that closeness of states $\widetilde{\rho}'$ and $\rho'$ in the trace norm implies closeness of the corresponding conditional mutual informations. First, we expand $I(X:B|\hat{T}^*)_{\widetilde{\rho}'}$ and $I(X:B|\hat{T}^*)_{\rho'}$ with respect to Bob, using mutual entropies: 
\begin{equation}
\label{close1}
\begin{split}
&\left|I(X:B|\hat{T}^*)_{\widetilde{\rho}'}-I(X:B|\hat{T}^*)_{\rho'}\right|=\left|S(B|\hat{T}^*)_{\widetilde{\rho}'}-S(B|\hat{T}^*)_{\rho'}\right|+\\
&+\left|S(B|X\hat{T}^*)_{\widetilde{\rho}'}-S(B|X\hat{T}^*)_{\rho'}\right|\leq 4\delta \log_2 d_B+2h(\delta).
\end{split}
\end{equation}
For functions $I(X:E|\hat{T}^*)_{\widetilde{\rho}'}$ and $I(X:E|\hat{T}^*)_{\rho'}$, we expand with respect to Eve's system, 
\begin{equation}
\label{close2}
\begin{split}
&\left|I(X:E|\hat{T}^*)_{\widetilde{\rho}'}-I(X:E|\hat{T}^*)_{\rho'}\right|\leq \left|S(E|\hat{T}^*)_{\widetilde{\rho}'}-S(E|\hat{T}^*)_{\rho'}\right|+\\
&+\left|S(E|X\hat{T}^*)_{\widetilde{\rho}'}+S(E|X\hat{T}^*)_{\rho'}\right|\leq 4\delta \log_2 d_E+2h(\delta).
\end{split}
\end{equation}
The state $\rho_{aABE}$ is pure, which implies that the dimension $d_E$ is upper bounded by $d_ad_Ad_B$. This follows from observation that, in the cut $aAB:E$, the Schmidt decomposition cannot have more terms than $\operatorname{rank}(\rho_{aAB})\leq \operatorname{dim}(\mathcal{H}_a\otimes \mathcal{H}_A\otimes \mathcal{H}_B)=d_ad_Ad_B$. This allows us to rewrite~\eqref{close2} as
\begin{equation}
\label{close2prim}
\begin{split}
&\left|I(X:E|\hat{T}^*)_{\widetilde{\rho}'}-I(X:E|\hat{T}^*)_{\rho'}\right|\leq 4\delta \log_2(d_ad_Ad_B)+2h(\delta)\\
&=4\delta \log_2 d_B +4\delta \log_2(d_ad_A)+2h(\delta).
\end{split}
\end{equation}
Now, combining expressions~\eqref{close1} and~\eqref{close2prim}, we get~\eqref{eq:close_final} finishing the proof.
\end{proof}

\begin{observation}
\label{dimX}
Inequality~\eqref{chain_thm1} in Theorem~\ref{thm:delta_bound} can be re-written in terms of dimension $d_X$ of the space of measurements $X$:
\begin{equation}
\label{eq:dimX}
K^{iid}(\rho_{ABE})\geq  K^{iid}(\rho_{aABE})-(8\delta \log_2 d_X+4h(\delta)),
\end{equation}
where $X$ is generated from $aA$ via iid measurement $\hat{Q}^*_x:aA\rightarrow X$.
\end{observation}
One can prove this statement by writing expressions~\eqref{close1} and \eqref{close2} with respect to space of outcomes $X$, and by similar lines as in the proof of Theorem~\ref{thm:delta_bound} one gets the statement. As we will see, this Observation is of the great importance when one considers private states with $d_X=d_A$, since considered measurements are the von Neumann measurements, which do not increase respective dimension. This significantly reduces the value of the factor in equation~\eqref{chain_thm1}.

We know that any Schmidt-twisted pure state $\widetilde{\gamma}_{ABA'B'}$ can be written  as $U(\psi_{AB}\otimes \sigma_{A'B'})U^{\dagger}$, where its explicit form is presented in~\eqref{Schmidt_pure}. In this class one can consider a subclass of irreducible Schmidt-twisted pure states. The whole secret content of these states is accessible via systems $A$ and $B$. Irreducible private states~\cite{irred_pbits} are a special case of these states.  An irreducible private state $\gamma$ with $2^k\otimes 2^k$ dimensional key part satisfies $K_D(\gamma) = k$. From Theorem~\ref{thm:delta_bound} and Observation~\ref{dimX} we have the following proposition: 

\begin{proposition}
\label{cor:main}
For an irreducible Schmidt-twisted pure states $\widetilde{\gamma}_{ABA'B'}$ with $AA'=aA''$,  there is
\begin{equation}
\label{Cor3eq}
        K_D(\widetilde{\gamma}_{A''BB'}) 
    \geq K_D(\widetilde{\gamma}_{aA''BB'})- [8 \delta \log_2 d_{A} + 4h(\delta)],
        \end{equation}
        with $\delta = \sqrt{1 - 2^{-I(a:BB'|A'')_{\gamma}}}$.
\end{proposition}

\begin{proof}
We apply the statement of Observation~\ref{dimX} to a pure state $\widetilde{\gamma}_{aA''BB'E}$ with measurement $Q_x^*$, which is composition of the unitary $U:aA''\rightarrow AA'$ with the von Neumann measurement on system $A$ in the computational basis, obtaining
\begin{equation}
K_D^{iid}(\widetilde{\gamma}_{A''BB'E})\geq K_D^{iid}(\psi_{\widetilde{\gamma}_{aA''BB'E}})- [8 \delta \log_2 d_{A} + 4h(\delta)].
\end{equation}
The following chain of equalities holds:
\begin{equation}
\label{eq:chain0}
\begin{split}
&K_D^{iid}(\psi_{\widetilde{\gamma}_{aA''BB'E}})=K_D^{iid}(\psi_{\widetilde{\gamma}_{ABA'B'E}})=\\
&C_D(\psi_{\widetilde{\gamma}_{ABA'B'E}})=K_D(\widetilde{\gamma}_{ABA'B'})=K_D(\widetilde{\gamma}_{aA''BB'}).
\end{split}
\end{equation}
The first equality holds since the unitary operation producing different cut of the Alice's systems $AA'\leftrightarrow aA''$ does not change the amount of the key. Furthermore, 
\begin{equation}
\label{ineq1}
K^{iid}(\psi_{\widetilde{\gamma}_{ABA'B'E}})\leq C_D(\psi_{\widetilde{\gamma}_{ABA'B'E}})=S(A)_{\psi},
\end{equation}
where $C_{D}$ denotes the rate of key distilled by means of LOPC operations, see subsection~\ref{sec:LOPC} and \cite{DevetakWinter-hash} ($\psi$ denotes $\psi_{\widetilde{\gamma}_{ABA'B'E}}$).
The inequality in (\ref{ineq1}) is obtained because we work with a restricted class of protocols, while the equality follows from the fact that from irreducible private state we obtain exactly $S(A)_{\psi}$ of the key. 
Next, we notice that, for irreducible Schmidt-twisted pure states, the inequality~\eqref{ineq1} is saturated, $K^{iid}(\psi_{\widetilde{\gamma}_{ABA'B'E}})=C_D(\psi_{\widetilde{\gamma}_{ABA'B'E}})$, because one achieves rate of $C_D(\psi_{\widetilde{\gamma}_{ABA'B'E}})=S(A)_{\psi}$ via the measurement, which is tensor power of the von Neumann  measurement on the key part $A$, while variable $T$ is null here (no communication is needed for obtaining the key). Due to~\cite{karol-PhD}, $C_D(\psi_{\widetilde{\gamma}_{ABA'B'E}})=K_D(\widetilde{\gamma}_{ABA'B'})$. Applying the unitary producing different cut of the Alice's systems $AA'\leftrightarrow aA''$, we obtain the last equality in~\eqref{eq:chain0}. To prove the left-hand side of~\eqref{Cor3eq} we observe that 
 \begin{equation}
 C_D(\psi_{\widetilde{\gamma}_{A''BB'E}})\geq K_D^{iid}(\widetilde{\gamma}_{A''BB'E}).
 \end{equation}
 Finally, using expression~\eqref{eq:CK} from subsection~\ref{sec:LOPC}, stating that, for a pure tripartite state $\psi_{ABE}$ with corresponding state $\rho_{AB}=\tr_E \psi_{ABE}$, one has $C_D(\psi_{ABE})=K_D(\rho_{AB})$,  we obtain the statement.
\end{proof}

It is tempting to ask how the bound from Theorem~\ref{thm:delta_bound} compares with the bound from
the Proposition~\ref{obs:shield_gen}. In Figure~\ref{fig:comp_bound1} we ask whether $8 \delta \log_2 d_{A} + 4h(\delta)\leq 2S(a)$, with $\delta = \sqrt{1 - 2^{-I(a:BB'|A'')_{\gamma}}}$.
\begin{figure}[h]
\centering
 \includegraphics[width=0.9\columnwidth,keepaspectratio,angle=0]{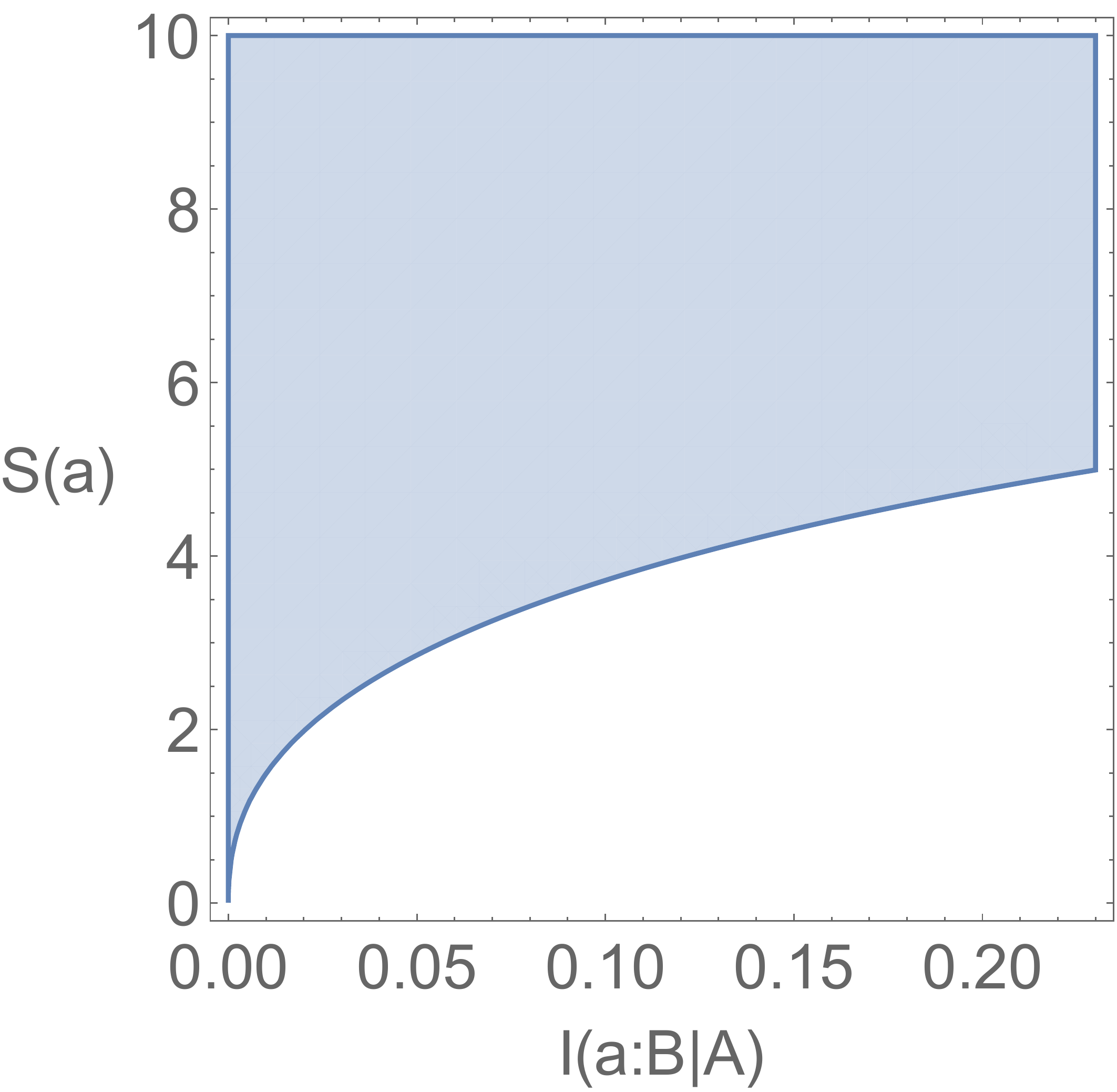}
\caption{A comparison of the bounds given in Theorem~\ref{thm:delta_bound}, with $\delta = \sqrt{1- 2^{-I(a:B|A)}}$, and in Observation~\ref{obs:shield}. The blue region corresponds to the case when $8 \delta \log_2 d_{A} + 4h(\delta)\leq 2S(a)$ with $d_{A}=2$. However, the quantities $S(a), I(a:B|A)$ are not independent, so not all pairs $(I(a:B|A),S(a))$ in the blue region are achievable. In other words, if a point is achievable, then it has to satisfy plotted relation, and otherwise we do not take it into account.} 
\label{fig:comp_bound1}
 \end{figure}

\section{partial non-locking for product of two states}
As we have mentioned earlier, it is an open problem if a two-way distillable key
drops down by more than $S(CD)$ upon the erasure of subsystems $CD$ of some bipartite state $\rho_{AC:BD}$. An easy sub-case of this problem is when the subsystem $CD$ is a {\it product} with the rest of the system $AB$. That is, we consider the consequences of the following transformation: 
\begin{equation}
\rho_{AB}\otimes \sigma_{CD} \rightarrow \rho_{AB}.
\end{equation}
It looks at first that the drop of a key
should be $K_D(\sigma_{CD})$. However, it need not be the case.
The problem that arises here stems from the fact that $K_D$ may be super-additive
on tensor product (this is known for the private capacity of quantum channels \cite{PrivCapSupAdd}). This is why it is not clear how much the key of $\rho_{AB}$
increases upon adding auxiliary system $\sigma_{CD}$. 

We argue
now that the increase can be controlled. 
\begin{observation} For a tensor product of biparite states $\rho_{AB}\otimes \rho_{CD}$, there is
\begin{eqnarray}
    &&K_D(\rho_{AB}\otimes \sigma_{AB}) -K_D(\rho_{AB}) \nonumber\\
    &&\leq \min\{E_R(\rho_{AB}),E_{sq}(\rho_{AB})\} - K_D(\rho_{AB}) + \nonumber \\
    &&\min\{S(\sigma_C), S(\sigma_D)\},
\end{eqnarray}
where 
$E_R(\rho) := \inf_{\sigma \in SEP}D(\rho,\sigma)$, with $D(\rho,\sigma):=\tr\rho\log_2\rho-\tr \rho \log_2 \sigma$, is the relative entropy of entanglement \cite{Vedral:Plenio:1998}, while $E_{sq}(\rho_{AB}):=\inf\{\frac{1}{2}I(A:B|E)\mid\rho_{AB}=\mathrm{tr}_E\rho_{ABE}\}$ is the squashed entanglement \cite{ChristandlWinter_squashed}.
\end{observation}
\begin{proof}
By noticing $K_D\leq \min \{E_R,E_{sq}\}$, we observe that
\begin{eqnarray}
&&K_D(\rho_{AB}\otimes \sigma_{AB}) -K_D(\rho_{AB}) \nonumber\\
&&\leq \min\{E_R(\rho_{AB}\otimes \sigma_{AB}),E_{sq}(\rho_{AB}\otimes \sigma_{AB})\}+\nonumber\\
&&-K_D(\rho_{AB}). 
\end{eqnarray}
We further note that $E_R$ is subadditive and $E_{sq}$ is additive on tensor product of the state. This leads to
\begin{eqnarray}
&&K_D(\rho_{AB}\otimes \sigma_{AB}) -K_D(\rho_{AB}) \nonumber\\
&&\leq \min\{E_R(\rho_{AB})+ E_R(\sigma_{AB}),E_{sq}(\rho_{AB})+ \nonumber\\ 
&&E_{sq}(\sigma_{AB})\} - K_D(\rho_{AB}). 
\end{eqnarray}
Finally, we have $\max\{E_R,E_{sq}\}\leq {E_C}$ where $E_C$ is an {\it entanglement cost } \cite{RMPK-quant-ent}, which satisfies ${E_C}\leq \min \{S(\sigma_C),S(\sigma_D)\}$.
\end{proof}

\begin{corollary}
	For a strictly irreducible private state $\gamma_{ABA'B'}$ and any state $\sigma_{CD}$,  
	there is $K_D(\gamma_{ABA'B'}\otimes \sigma_{CD})- K_D(\gamma_{ABA'B'})\leq \min \{S(\sigma_C),S(\sigma_D)\}$.
	\label{cor:strictly}
\end{corollary}
\begin{proof}
Follows from the fact that strictly irreducible private states satisfy
$E_R(\gamma_{ABA'B'}) = K_D(\gamma_{ABA'B'})$ \cite{irred_pbits}.
\end{proof}
We note, that similar corollary holds 
for the {\it maximally correlated states} of the form $\sum_{i,j} b_{ij}|ij\>
\<ij|$. For these states $E_D = K_D = E_R$ \cite{RMPK-quant-ent}. 

The system $CD$ can be viewed as a subsystem of the shield $A'B'$. In that case, Observation \ref{obs:shield} applies.  The above bound is tighter than the latter one, however it holds for a subclass of private states, and for a special case in which system $CD$ is a product with $ABA'B'$. 

Furthermore, the bound given in Theorem \ref{thm:delta_bound} applies in this case with $\delta = \sqrt{1- 2^{- I(C:BB'D|AA')}}=\sqrt{1- 2^{-I(C:D)}}$. In Figure~\ref{fig:comp_bound2} we compare the range of applicability of the latter bound with the one given in Corollary \ref{cor:strictly}.
\begin{figure}[h]
\centering
 \includegraphics[width=0.9\columnwidth,keepaspectratio,angle=0]{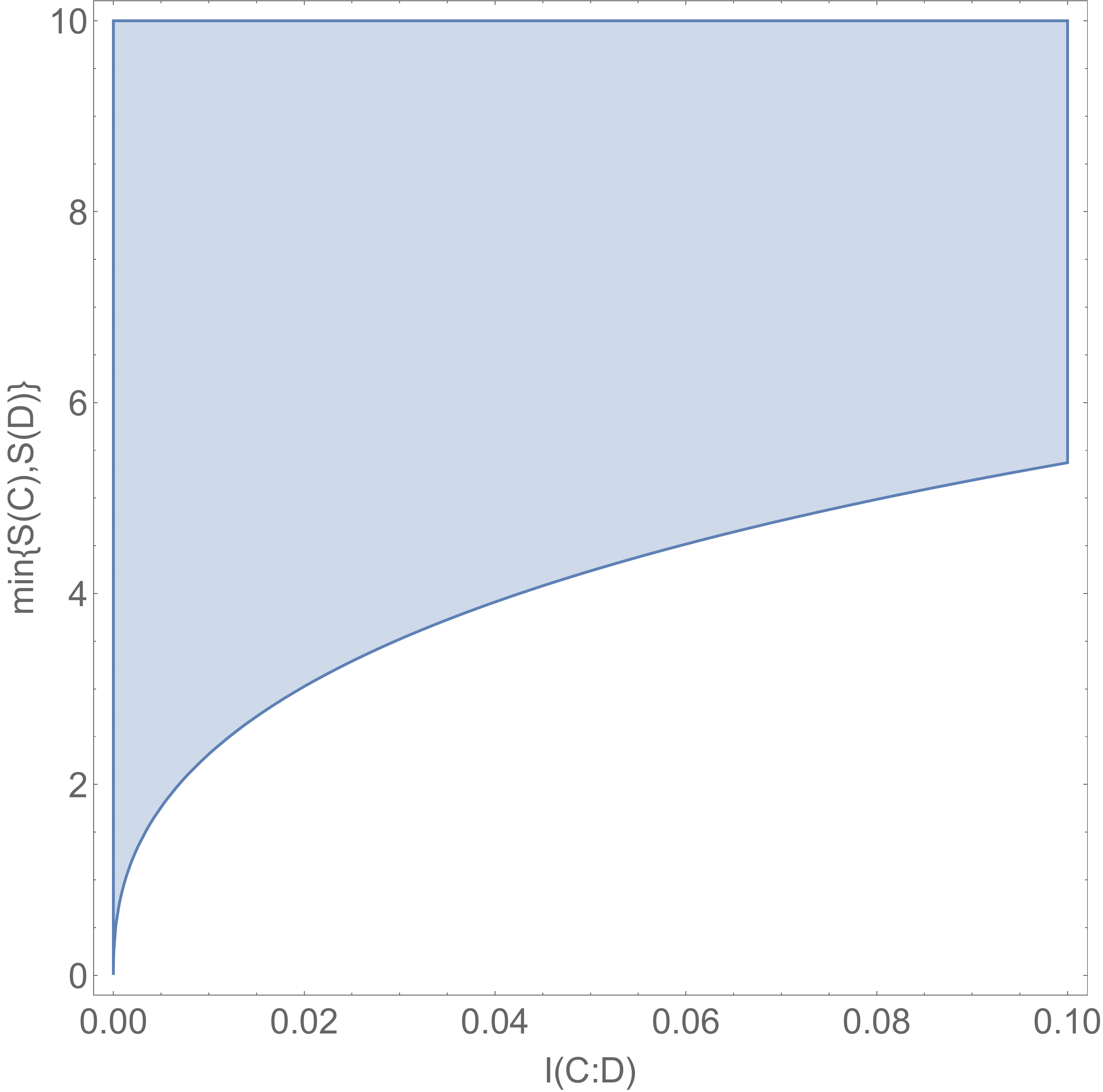}
\caption{A comparison of the bounds given in Theorem~\ref{thm:delta_bound}, with $\delta = \sqrt{1- 2^{-I(C:D)}}$, and through Corollary~\ref{cor:strictly}. The blue region corresponds to the case when $8 \delta \log_2 d_{A} + 4h(\delta)\leq \min \{S(\sigma_C),S(\sigma_D)\}$ with $d_{A}=2$.}
\label{fig:comp_bound2}
 \end{figure}

We now propose a weaker, but more general bound. 
\begin{observation}
	For a bipartite state $\rho_{A:BC}$ there is:
	\begin{eqnarray}
	K_D(\rho_{A:BC})-K_D(\rho_{AB}) \leq \nonumber\\I(A:C|B)+E_R^{\infty}(\rho_{AB})- K_D(\rho_{AB}),
	\end{eqnarray}
	where $E_R^\infty (\rho) := \lim_{n\rightarrow \infty} {1\over n}E_R(\rho^{\otimes n})$.
\end{observation}
\begin{proof}
We upper bound $K_D(\rho_{A:BC})$ by
$E_R^{\infty}(\rho_{A:BC})$ \cite{pptkey,keyhuge}. We then add and subtract $E_R^{\infty}(\rho_{AB})$. Lemma $1$ of \cite{Brando2011,Brando2012Erratum} allows
to upper bound the difference $E_R^{\infty}(\rho_{A:BC}) - E_R^{\infty}(\rho_{A:B})$ by $I(A:C|B)$, which proves the thesis.
\end{proof}

As an immediate corollary, we have that, for the state $\rho_{A:BC}$ such that the leftover state satisfies $E_R(\rho_{AB})=K_D(\rho_{AB})$, the upper bound on the loss of key is $I(A:C|B)$.

\section{Examples of action of side channels for some private states}
A motivation for this and the next section is given by the fact that certain private states, as well as states with a positive partial transposition that approximate them, are candidates for the {\it hybrid quantum network} design \cite{Sakarya2020}. This design ensures that unauthorized key generation will be impossible in quantum networks. It is therefore important to know how the distillable key of the latter states behaves under specific side channels. 

The findings of Section \ref{sec:lb_for_generated_key} ensure us that, 
upon the erasure of a single qubit of the shield (and hence upon any channel on it), the distillable key of a private state does not decrease by more than twice the entropy of the qubit (see  Proposition~\ref{obs:shield_gen}). In this Section, we concentrate on upper bounds on the drop of a key.
Namely, we consider special private states and channels and show the behavior of a key under the latter. 

The main result of this Section is an observation that the action on just one qubit of the shield of a certain private state can decrease the key by  half, irrespectively of the dimension of the shield (which varies in some range). This means that the protection of the state is not a monotonically increasing function of the number of qubits in the shield.

We  consider attacks on state $\gamma_V$,
given by \eqref{eq:X_form} with $X = V = \frac{1}{2d_s^2}\sum_{i=0,j=0}^{d_s-1}|ij\>\<ji|$ being the (normalised to half) swap operator. Specifically, we consider three values of local dimension of the shield: $d_s = 2,4,8$, and an attack by the {\it bit-flip} channel, specified as an operation $\Lambda_{bf}(\rho):= \alpha (\sigma_x^{A'}\otimes \mathbf{1}_{ABB'})\rho(\sigma_x^{A'}\otimes\mathbf{1}_{ABB'}) + (1-\alpha)\rho$, where $\sigma_x^{A'}$ is the Pauli matrix applied to system $A'$. We upper bound the value of key by $E_R(\rho)$ \cite{pptkey}. As a specific state $\sigma$ we choose the state $(1-p)\sigma_{att} + p {\mathbf{1}\over (2d_s)^2}$, where $\sigma_{att} = \Lambda_{bf}(\frac{1}{2}(|00\>\<00|\otimes\frac{\mathbf{1}}{d_s}+|11\>\<11|\otimes\frac{\mathbf{1}}{d_s}))$. The minimal value of an upper bound reached by this operation reads $0.5$. The result is shown on Fig. \ref{fig:bit_flip}. 

\begin{figure}[h]
\centering
 \includegraphics[width=1.0\columnwidth,keepaspectratio,angle=0]{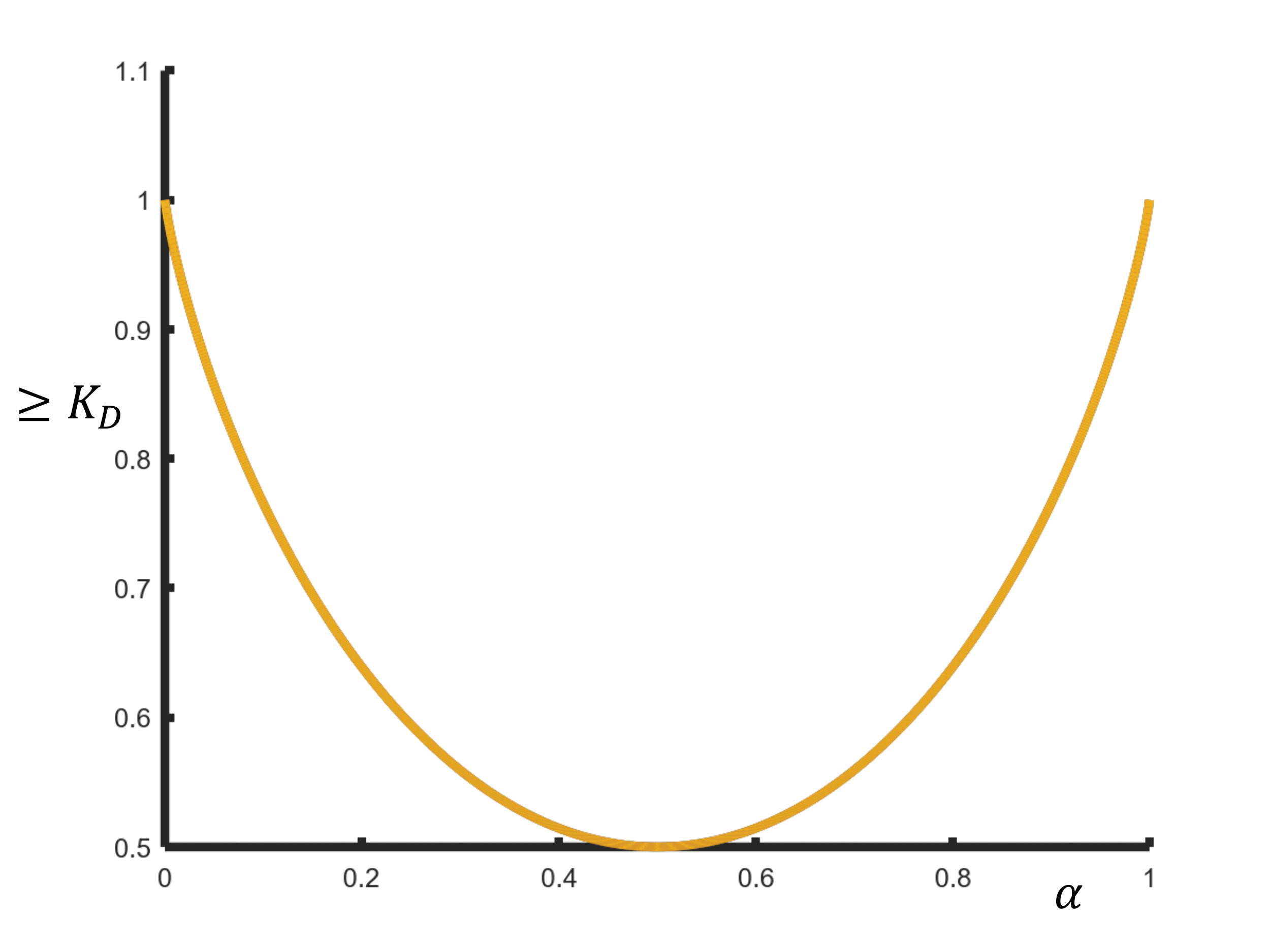}
\caption{Upper bound on the relative entropy of entanglement (and hence on $K_D$) of the state $\gamma_V$, after acting with the bit-flip channel on a qubit of its shield. The same plot is obtained for $d_s=2,4,8$, hence larger shield is no more shielding than smaller one.} 
\label{fig:bit_flip}
 \end{figure}

For the same state, we consider the action of {\it depolarizing} channel, specified by
\begin{equation}
    \Lambda_{dep}(\cdot) = (1-\frac{3\alpha}{4})\mathbf{1}(\cdot) + \frac{\alpha}{4}\sigma_x(\cdot)\sigma_x +\frac{\alpha}{4}\sigma_y(\cdot)\sigma_y+
    \frac{\alpha}{4}\sigma_z(\cdot)\sigma_z. 
\end{equation}
The maximal drop of the relative entropy of entanglement (and hence the key) reads $0.18872$, for $\alpha =1$. Resulting plot is depicted on Fig. \ref{fig:depolarising}.

\begin{figure}[h]
\centering
 \includegraphics[width=1.0\columnwidth,keepaspectratio,angle=0]{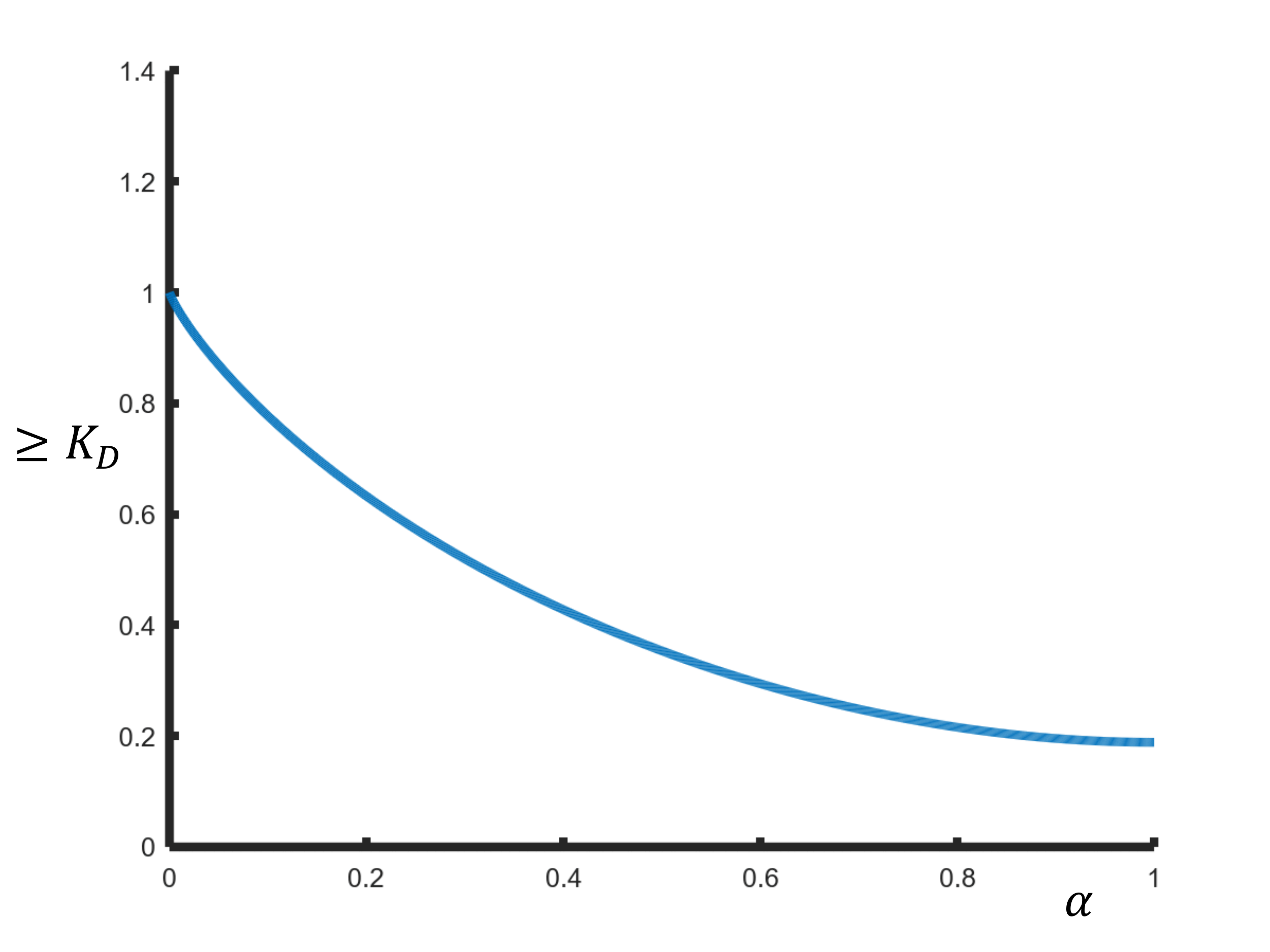}
\caption{Upper bound on the relative entropy of entanglement of the state $\gamma_V$ (and hence $K_D$), after acting with depolarizing channel on a qubit of its shield. The same plot is obtained for $d_s=2,4,8$.} 
\label{fig:depolarising}
 \end{figure}

Next, we check the action of the {\it amplitude damping} channel, ${\cal N}_\alpha(\cdot) = M_1(\alpha)(\cdot)M_1(\alpha)^{\dagger} +M_2(\alpha)(\cdot)M_2(\alpha)^{\dagger}$, which is specified by parameter $\alpha \in [0,1]$ and the following two Kraus operators:
\begin{equation}
    M_1(\alpha) = 
    \left[\begin{array}{cc}
    1 & 0 \\
    0 & \sqrt{1-\alpha}
    \end{array}\right], 
    M_2(\alpha) = 
    \left[\begin{array}{cc}
    1 & \sqrt{\alpha} \\
    0 & 0
    \end{array}\right].
\end{equation}
The minimal value reached in this case is also 
$0.18872$, and the results are the same for $d_s=2,4,8$. They are plotted on Fig. \ref{fig:amp_dump_channel}.

\begin{figure}[h]
\centering
 \includegraphics[width=1.0\columnwidth,keepaspectratio,angle=0]{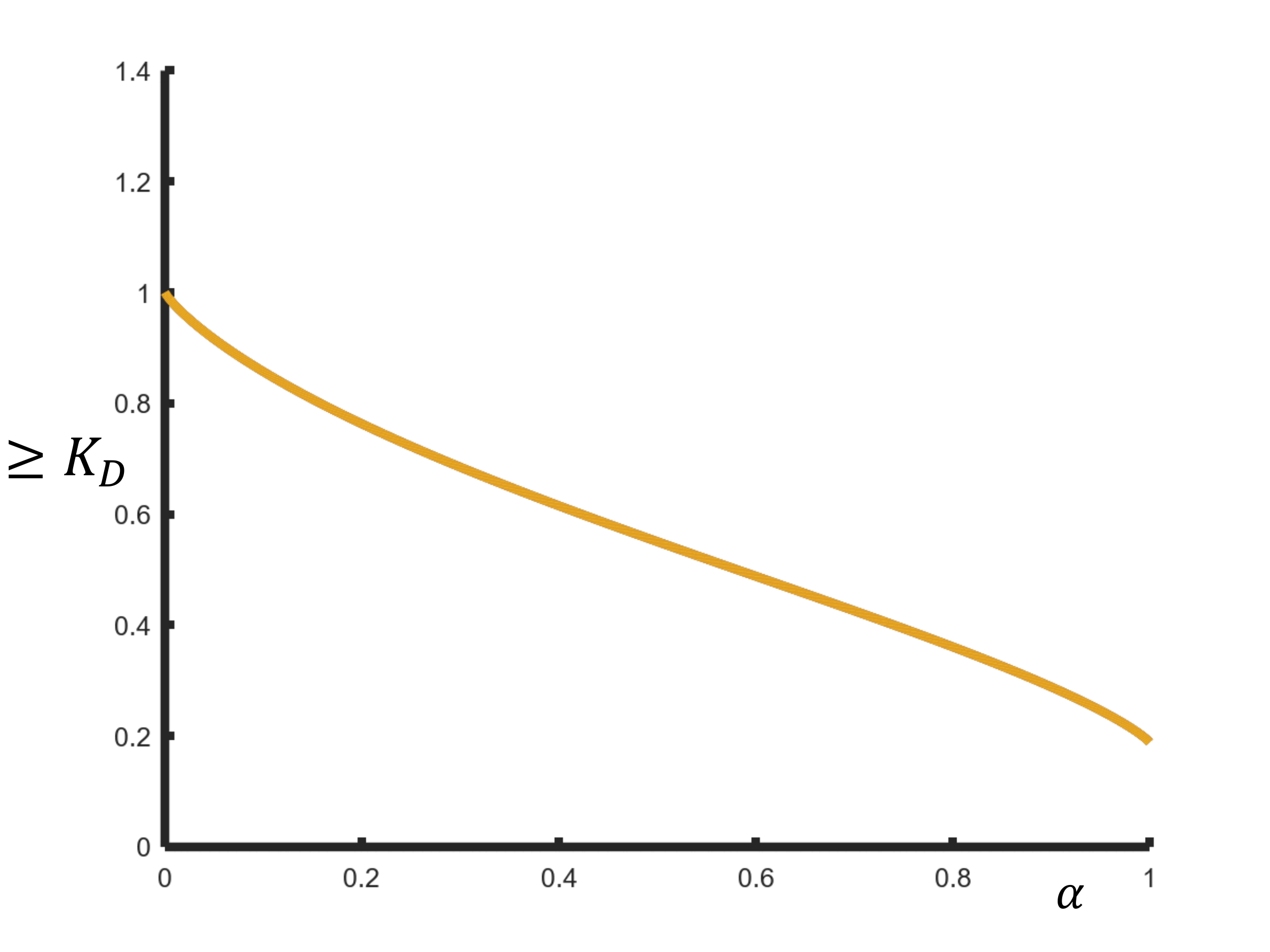}
\caption{Upper bound on the relative entropy of entanglement of the state $\gamma_V$ (and hence $K_D$), after acting with amplitude damping channel on a qubit of its shield. The same plot is obtained for $d_s=2,4,8$.} 
\label{fig:amp_dump_channel}
 \end{figure}

\section{Connection of leakage with the non-markovianity of dynamics}
In this section, we reveal the connection between the problem of (non)markovianity of a quantum dynamics and that of hacking.
We will see that a dynamics is markovian, then 
for all block states, their key witnessed by certain non-linear privacy witness does not increase in time under the dynamics.

Given a family $\{\Lambda_t\mid t\geq0\}$ of CPTP maps (interpreted as a temporal dynamics of a system), there is a range of different (generally inequivalent) conditions that can be imposed on this family, to make it called (by, generally, different authors) a 'quantum markovian dynamics' (see, e.g., \cite{Wolf:Cirac:2008,Li:Hall:Wiseman:2018} for review and comparison). Among those conditions, \textit{CP-divisibility}, introduced in \cite{Wolf:Cirac:2008} and futher studied in \cite{Rivas:Huelga:Plenio:2010}, is defined as existence of a CPTP map $V_{t,s}$ such that $\Lambda_t=V_{t,s}\Lambda_s$ $\forall t\geq s$. In this paper we fix a terminological choice, identifying markovianity with CP-divisibility.

In what follows we will construct an analogue of a recent result by Kołodyński et al. \cite{Kolod2020}, who found that an entanglement measure known as {\it negativity} is an indicator of nonmarkovianity. The authors of \cite{Kolod2020} provide examples of tripartite states
and show that the invertible map is nonmarkovian if{}f there exist a  specially designed tripartite state whose negativity increases in time. (The \textit{invertibility} of $\Lambda_t$ is understood everywhere here as left invertibility, i.e. $\exists!$ $\Lambda_t^{-1}$ such that $\Lambda_t^{-1}\circ\Lambda_t=\mathbf{1}$.) More precisely, in \cite{Kolod2020} there were considered block states 
of the form (using notation of the latter paper):
\begin{eqnarray}
    \tau_t^{ABC} &=& p_1 (\Lambda_t^A\otimes \mathbf{1}^{B_1})(\rho_1^{AB_1}) \otimes |\psi_+\>\<\psi_+|^{B_2C}+ \nonumber\\&+& p_2(\Lambda_t^A\otimes \mathbf{1}^{B_1})(\rho_2^{AB_1}) \otimes |\psi_-\>\<\psi_-|^{B_2C},
\end{eqnarray}
where $\rho^{ABC}:=\tau^{ABC}_{t=0}$ for $\Lambda_{t=0}=\mathbf{1}$. It is shown there that the negativity $E_N$ \cite{yczkowski1998,EisertPhd,VidalWerner_negativity},  
computed in the cut $C: B_1B_2A$, witnesses nonmarkovianity of 
dynamics. 
\begin{theorem}[Theorem 2 of \cite{Kolod2020}]
    For any invertible nonmarkovian evolution $\{\Lambda_t\mid t\geq0\}$
there exists a quantum state $\rho_{ABC}$ such that
\begin{equation}
    \frac{d}{dt}{E_N}^{AB|C}(\tau^{ABC}_t) >0
\end{equation}
for some $t > 0$. For single-qubit evolutions $\Lambda_t$ the statement
also holds for non-invertible dynamics.
\end{theorem}

We observe that these states, treated as {\it bipartite}, are block states, and in special cases also private states. This motivates us to study the connection between the topic of privacy and nonmarkovianity.

The proof of a result of \cite{Kolod2020} is based on a theorem in \cite{Chruciski2011}, which states that CP-divisibility for a family $\{\Lambda_t\mid t\geq0\}$ of \textit{invertible} CPTP maps is equivalent to a condition $\frac{d}{dt}||(\Lambda_t\otimes \mathbf{1})X||_1\leq0$ $\forall X\in\mathcal{B}(\mathcal{H})\otimes\mathcal{B}(\mathcal{H})$ with $X=X^\dagger$, and $\mathcal{B}(\mathcal{H})$ denotes space of all bounded operators on $\mathcal{H}$. In \cite{Chruscinski:Rivas:Stoermer:2018} this result has been extended to noninvertible families of CPTP maps satisfying $\mathrm{im}(\Lambda_t)\subseteq\mathrm{im}(\Lambda_s)$ $\forall t>s$ (i.e., \textit{image nonincreasing}), and we will use this extension below.
 
In what follows, we first show the behavior
of the privacy witness under an attack of a 
hacker. Hacker acts on the system $A'$,
and her attack is represented by operation
$\Lambda_{A'}$. As we will see, the 
privacy witness degrades monotonically with
the decrease of $||(\Lambda_{A'}\otimes {\mathbf{1}}_{B'}){1\over 2}(p_+\rho_+-p_-\rho_-)||_1$.

\begin{proposition}[Nonlinear privacy witness]
	Let $\rho_{ABA'B'} = p_+ |\psi_+\>\<\psi_+|_{AB}\otimes \rho_+^{A'B'} +p_- |\psi_-\>\<\psi_-|_{AB}\otimes \rho_-^{A'B'}$, $\Lambda_{A'}$ a CPTP map acting on system $A'$ of $\rho_{ABA'B'}$, $[\Lambda(\rho)]_{psq}$ be the privacy-squeezed state of $\Lambda(\rho)=\Lambda_{A'}\otimes {\mathbf{1}}_{ABB'}(\rho_{ABA'B'})$ %$(\Lambda_{A'}\otimes \mathbf{1}_{B'})X$, 
	. Then:
	\begin{equation}
	 K_D([\Lambda(\rho)]_{psq})= 1 - h\left({1\over 2}+||(\Lambda_{A'}\otimes \mathbf{1}_{B'})X||_1\right),
	\end{equation}
	where $X={1\over 2}(p_+\rho^{A'B'}_+-p_-\rho^{A'B'}_-)$.
	\label{prop:Markov}
\end{proposition}
\begin{proof}
 For the first inequality, we upper bound the amount of key of $[\Lambda(\rho)]_{psq}$ via the relative entropy of entanglement. We note that the state under consideration is 
Bell-diagonal, of the form  $q_+|\psi_+\>\<\psi_+| + q_-|\psi_-\>\<\psi_-|$. Thus, its relative entropy of entanglement
reads $1 - h(p_{max})$, where $p_{max}$ is the maximal probability of a Bell state in the mixture \cite{Vedral1997}. In our case $\frac{1}{2}(q_+ - q_-) = ||(\Lambda_{A'} \otimes \mathbf{1}_{B'})X||_1=:c$, hence $q_+={1\over 2}+c$ and $q_-={1\over 2}-c$.
Since $c\geq 0$, $q_+\geq q_-$, and so
\begin{equation}
K_D([\Lambda(\rho)]_{psq})\leq E_R([\Lambda(\rho)]_{psq}) = 1-h\left({1\over 2}+c\right).
\end{equation}
To see the lower bound we note that
\begin{equation}
    K_{DW}([\Lambda(\rho)]_{psq}) \leq K_D([\Lambda(\rho)]_{psq}),
\end{equation}
where $K_{DW}$ is the rate of Devetak--Winter protocol \cite{DevetakWinter-hash}.
The lower bound follows then from Corollary $1$ of \cite{CFH}, which states that 
$K_{DW}(\rho_{psq}) \geq 1-H(\alpha+\gamma,\alpha-\gamma,\beta,\beta)$, where $\alpha = (p_++p_-)/2=1/2$ , $\beta= 0$ and $\gamma=||(\Lambda_{A'}\otimes \mathbf{1}_{B'})X||_1$. Hence the assertion follows.
\end{proof}
Hence, the key of privacy squeezed state of an $\rho$ attacked
by $\Lambda_{A'}$ is a privacy witness of $\Lambda_{A'}(\rho)$, and
is monotonically strictly decreasing with the decrease of $||(\Lambda_{A'}\otimes \mathbf{1}_{B'})X||_1 \in [0,{1\over 2}]$, for hermitian $X$ representing the state.

To uncover the connection between hacking and (non)markovianity, we observe that: 
\begin{itemize}
\item The rate of any protocol of key distillation from a quantum state $\rho$ quantifies the resource
(how much key can be gained from a given state).
Hence, as the time passes, it can only stay the same (e.g., as a result of local unitary transformation on $\rho$),
or decrease (e.g., as a result  of the action of the local partial trace of a subsystem of $\rho$). 
\item The (invertible or image nonincreasing) dynamics $\{\Lambda_t\mid t\geq0\}$ (acting on the system $A'$) is markovian if{}f the map $\Lambda_t\otimes \mathbf{1}_{B'}$ either preserves the the trace norm of $X$ or decreases it for all hermitian $X$ and all $t>0$ \cite{Chruciski2011,Chruscinski:Rivas:Stoermer:2018}. 
\end{itemize}

Using Proposition \ref{prop:Markov} and equality of dimensions of $A'$ and $B'$, we can formulate an analogue of Theorem 2 of \cite{Kolod2020}.
\begin{theorem}
	An invertible or image nonincreasing dynamics $\{\Lambda_t\mid t\geq0\}$ is nonmarkovian if{}f there exists a block state \eqref{rhoblock} and $t>0$ such that
	\begin{equation}
	    \frac{d}{dt}K_D\left([\Lambda_t(\rho)]_{psq}\right)>0.
	\end{equation}
\end{theorem}
\begin{proof}
From \cite{Chruciski2011,Chruscinski:Rivas:Stoermer:2018} we have an equivalence of CP-divisibility with $\frac{d}{dt}||(\Lambda_t\otimes \mathbf{1}_{B'})X||_1\leq0$ for all $X$ and all $t>0$. This, combined with equivalence of $\frac{d}{dt}||(\Lambda_t\otimes \mathbf{1}_{B'})X||_1>0$ with $\frac{d}{dt}h\left(\frac{1}{2}+||(\Lambda_t\otimes \mathbf{1}_{B'})X||_1\right)<0$, and with  Proposition~\ref{prop:Markov}, completes the proof. 
\end{proof}

The above theorem establishes a link with an operational quantity, the {\it witnessed distillable key} (WDK), rather than with a theoretical measure
of entanglement, such as the negativity $E_N$.
It can be interpreted as follows: nonmarkovian dynamics implies the flow of privacy from environment to the system.

Interestingly, WDK is not an entanglement measure. Indeed, to make WDK zero for a block state, it is enough that $||p_+\rho_+ -p_-\rho_-||_1 = 0$, which
is true for $X = p_+\rho_+ - p_-\rho_-$ being a zero matrix. This implies $p_+=p_-={1\over 2}$ and $\rho_+=\rho_-\equiv \rho$. In this case the block state takes form $\frac{1}{2}(|00\>\<00|+|11\>\<11|) \otimes \rho$. However, if $\rho$ is entangled, then WDK is zero, while the block state is clearly entangled as a product of separable and entangled state. It would be interesting to extend this result to other operational entanglement measures, possibly via the approach of  \cite{FerraraChristandl}.
 Finally we note, that WDK is the inherently non-linear witness of nonmarkovianity. In that, this approach is complementary to that of considered earlier in \cite{banacki2020information}, where {\it linear} witness of a slightly different notion of nonmarkovianity, has been proposed.

\section{Discussion}
We have provided bounds on the leakage of private randomness and private key. We have shown that the private randomness in distributed setting can not drop down by more than $S(a)+\log_2 |a|$ upon unitary transformation followed by the erasure of a system $a$. It would be interesting to consider a more general case, in which a POVM is performed by 
the hacker. In this case, the difficulty  rests in controlling the amount of private randomness that can be added to the system. Indeed, every POVM can be considered as von Neumann measurement on the embedded system. However, embedding implies attaching a pure state, i.e., the state with private randomness, which we would like to avoid in the resource-theoretic approach.

Regarding private key, we have proved its non-lockability for the first non-trivial class of mixed states - the class of irreducible private states. Let us note here that the assumption that the state is irreducible is not restrictive. Indeed, a non-irreducible private state can have an arbitrary state on the shield. Hence non-locking for the general private state is as hard as the still open problem of non-lockability of the key for any state. We have shown that the bound on leakage (that reads $2S(a)$) is tight. We then provided a refinement of this result, which reflects the fact, that less
correlated qubits affect the drop of key by less amount, dependent on the value of $I(a:B|A)$. We have
done it for generalized private states called {\it irreducible Schmid-twisted pure states}. It is an interesting open problem
if the same would hold for the class of twisted pure states.  Another open problem which arises concerns one-way distillable key by means of communication from $A$ to $B$. Our upper bounds for the leakage differ in the case
when the leakage affects the system $A$ and from the case when it affects system $B$. It is an open problem if they need to differ, that is whether one-way distillable key from $A$ to $B$ drops down by a different number for some state when the same leaking channel acts on system $A$ from the case when it acts on system $B$.

We have also considered the effect of the leakage via exemplary side channels. For the considered private state, we observed that
the key drops down by the same amount irrespectively of the size of
the shield. This means that it is not the case that the larger is shield, the more protected is the key of this private states.
Designing private states which are immune to the qubit loss on the shield (and having low distillable entanglement) would be a good step towards the hybrid quantum network provided in \cite{Sakarya2020}.

Still, however, a major theoretical problem rests in answering the question of how much the key drops down under the erasure of a system of an arbitrary quantum state.
As we argue, it remains open even in 
the case when the system is in tensor product
with the rest of the state under consideration. 

Finally, we proved a connection between
the (non)markovianity of quantum dynamics and hacking. We have found  an operational quantity which is a non-linear private key witness, $K_D([\rho]_{{psq}})$, the key of a privacy-squeezed state.
In this context, it would be interesting to find an operational entanglement measure, the behavior of which corresponds to (non)markovianity 
of dynamics. It is also interesting if other variants of the  definition of (non)markovianity
can be connected to a secret key extraction (see \cite{banacki2020information} in this context).

\section*{Acknowledgements}
KH acknowledges Marcus Grassl and Anindita Bera for enlightening discussions. 
KH, OS, and RPK acknowledge  support by the National Science Centre grant Sonata Bis 5 (grant number: 2015/18/E/ST2/00327). DY was supported by the NSFC (grant nos 11375165, 11875244), and by the NFR Project ES564777. We acknowledge partial support by the Foundation for Polish Science (IRAP project, ICTQT, contract no. MAB/2018/5, co-financed by EU within Smart
Growth Operational Programme). The ’International Centre for Theory of Quantum Technologies’ project (contract no. MAB/2018/5) is carried out within the International Research Agendas Programme of the Foundation for Polish Science co-financed by the European Union from the funds of
the Smart Growth Operational Programme, axis IV: Increasing the research potential (Measure 4.3). 

\section{Appendix}
Here we partially recover Lemma V.3 of \cite{Nowakowski2016}. The problem with
the original statement of this lemma is: when two states $\rho_{AB}$ and $\widetilde{\rho}_{AB}$ are close in trace norm, it does not imply that the state $\rho_{AE}$ and $\widetilde{\rho}_{AE}$ are so (here $\rho_{ABE}$ is an extension of
$\rho_{AB}$ to system $E$). However this 
holds true, yet with a worse factor,
given the extension $\rho_{ABE}$ is pure.

In what follows, we use the fidelity defined by $F(\rho,
\sigma):= ||\sqrt{\rho}\sqrt{\sigma}||_1^2$.
Before showing a restatement of the aforementioned lemma, we show that if two bipartite states are close, so are their purifications (this technique was used before in \cite{keyhuge}, which we recall here for the completeness of the presentation). 

Let $||\rho_{AB}-\widetilde{\rho}_{AB}||\leq \delta$. By the Fuchs--van de Graaf inequality \cite{Fuchs:vandeGraaf:1999} we have
\begin{equation}
    \sqrt{F(\rho_{AB},\widetilde{\rho}_{AB})} \geq 1 - \frac{\delta}{2}.
\end{equation}
On the other hand, by the Uhlmann theorem \cite{Uhlmann:1976},
$F(\rho_{AB},\widetilde{\rho}_{AB})= \max_{\phi_{\widetilde{\rho}_{AB}}} |\<\psi_{\rho_{AB}}|\phi_{\widetilde{\rho}_{AB}}\>
|^2$ and 
$|\<\psi_{\rho_{AB}}|\phi_{\widetilde{\rho}_{AB}}\>|^2= F(\psi_{\rho_{AB}},\phi_{\widetilde{\rho}_{AB}})$, 
where $\psi_{\rho_{AB}}$ 
and $\phi_{\widetilde{\rho}_{AB}}$ are purifications of $\rho_{AB}$ and $\widetilde{\rho}_{AB}$ respectively. Applying again the Fuchs--van de Graaf inequality, we obtain:
\begin{align}
&|| |\psi_{\rho_{AB}}\>\<\psi_{\rho_{AB}}| - |\phi_{\widetilde{\rho}_{AB}}\>\<\phi_{\widetilde{\rho}_{AB}}|||_1 \leq \nonumber\\
&\sqrt{1 - \left(1- \frac{\delta}{2}\right)^2}\leq \sqrt{2\delta}.
\label{eq:closeby}
\end{align}
Lemma \ref{lem:biparite} (below) recovers the content of Lemma V.3 of \cite{Nowakowski2016} for the case of system $E$ purifying systems $AB$.
(By notation $K^{\rightarrow}(\rho_{AB})$ we  mean
$K^{\rightarrow}(|\psi_{\rho_{AB}}\>)$, where $\tr_{E}|\psi_{\rho_{ABE}}\>\<\psi_{\rho_{ABE}}|=\rho_{AB}$.) 
\begin{lemma}
\label{lem:biparite}
For bipartite states $\rho_{AB}$ and $\widetilde{\rho}_{AB}$ satisfying $||\rho_{AB}-\widetilde{\rho}_{AB}||_1\leq \delta$ with $\delta \leq \frac 12$, there is
\begin{align}
   & |K^{\rightarrow}(\rho_{AB}) - K^{\rightarrow}(\widetilde{\rho}_{AB})|\leq  (4\delta + 4\sqrt{2\delta})\log_2 d_A +\nonumber\\ &2h(\delta)+2 h(\sqrt{2\delta}). 
\end{align}    
\end{lemma}
\begin{proof}
Following \cite{Nowakowski2016}, we consider difference of conditional entropies: $K^{\rightarrow}(\rho)=-S(A|BT) + S(A|ET)$, where $T$ is generated via measurement on system $A$. Hence, \begin{equation}
    ||\rho_{ABT}-\widetilde{\rho}_{ABT}||_1\leq \delta,
\end{equation} since the trace norm does not increase under CPTP maps. Further, from \eqref{eq:closeby}, there is 
$||\rho_{AE} -\widetilde{\rho}_{AE}||_1 \leq \sqrt{2\delta}$ and, by the same argument, 
\begin{equation}||\rho_{AET} -\widetilde{\rho}_{AET}||_1 \leq \sqrt{2\delta}.
\end{equation} 
We have then
\begin{align}
   & |K^{\rightarrow}(\rho_{AB}) - K^{\rightarrow}(\widetilde{\rho}_{AB})|\leq |S(\widetilde{A}|\widetilde{BT}) - S(A|BT)|+ \nonumber\\ &|S(A|ET)- S(\widetilde{A}|\widetilde{ET})|.
   \label{eq:keys_close}
\end{align}
We further bound the two terms in r.h.s.
using Theorem by Alicki and Fannes \cite{Alicki-Fannes}, which states that if two states $\rho_{AB}$ and $\sigma_{AB}$ satisfy $\epsilon = ||\rho_{AB}-\sigma_{AB}||_1$, then
\begin{equation}
    |S(A|B) - S(\widetilde{A}|\widetilde{B})| \leq 4\epsilon \log_2 d_A + 2 h(\epsilon),
\end{equation}
where $d_A$ is dimension of system $A$ and
$h(\cdot)$ is the binary Shannon entropy.
Applying the above inequality to (\ref{eq:keys_close}), we obtain
\begin{align}
   & |K^{\rightarrow}(\rho_{AB}) - K^{\rightarrow}(\widetilde{\rho}_{AB})|\leq + 4\delta \log_2 d_A + 2 h(\delta) \nonumber\\ & 4\sqrt{2\delta} \log_2 d_A + 2 h(\sqrt{2\delta}),
\end{align}
only if $\sqrt{2\delta} \leq {1\over 2}$, and hence $\delta \leq {1\over 2}$. Here we use the fact that $h(x)$ is strictly increasing for $x \in [0,\frac 12]$, so that $h(||\rho_{AB}-\widetilde{\rho}_{AB}||_1) \leq h(\sqrt{2\delta})$.
\end{proof}
It is important to note that quantum purification is the worst extension from
the cryptographic point of view because it allows an eavesdropper to
create any other extension by local operation. Hence, the above result is important from a cryptographic point of view.
\bibliography{references}
 \end{document}